\documentclass[11pt]{article}

\usepackage[final]{acl}

\usepackage{times}
\usepackage{latexsym}

\usepackage[T1]{fontenc}

\usepackage[utf8]{inputenc}

\usepackage{microtype}

\usepackage{inconsolata}

\usepackage{graphicx}

\usepackage{tikz}
\usepackage{arydshln}
\usepackage{algpseudocode}
\algrenewcommand\textproc{\text}
\usepackage{makecell}
\usepackage{booktabs}
\usepackage{color}
\usepackage{colortbl}
\usepackage{multirow}
\usepackage{enumitem}
\usepackage{makecell}
\usepackage{threeparttable}
\usepackage{amsmath,amsfonts,mathtools}
\usepackage{pifont}

\usepackage{float}
\usepackage{graphicx}
\usepackage{xcolor}
\usepackage{enumitem}
\usepackage{tabularx, booktabs}

\usepackage{tcolorbox}

\usepackage{colortbl}
\usepackage{float}
\usepackage{graphicx}
\usepackage{hyperref}
\usepackage{tabularx, booktabs}
\usepackage{tikz}
\usepackage{arydshln}
\usepackage{CJKutf8}

\usepackage{bm}        
\usepackage{amsfonts}
\usepackage{cancel}

\definecolor{mygreen}{HTML}{00B050}

\definecolor{myorange}{HTML}{ED7D31}
\definecolor{rowgray}{gray}{0.97}
\definecolor{avgblue}{RGB}{210,230,250}  
\definecolor{headergray}{RGB}{160,160,160} 
\definecolor{uc_color}{rgb}{0.99,0.24,0.63}
\definecolor{hc_color}{rgb}{0.02,0.51,0.51}
\definecolor{tc_color}{rgb}{0.99,0.55,0.09}
\definecolor{posgreen}{RGB}{0,150,0}
\definecolor{negred}{RGB}{200,0,0}

\usepackage{array}
\usepackage{mathtools}
\usepackage{multirow}
\usepackage{subcaption}
\usepackage{tikz}
\renewcommand{\arraystretch}{1.1}
\definecolor{color1}{cmyk}{0.216,0.176,0,0}
\definecolor{color2}{cmyk}{0.059,0.235,0.392,0}

\usepackage{graphicx}
\usepackage{tikz}
\usepackage{wrapfig}
\usepackage{algorithm}
\usepackage{algpseudocode}
\algrenewcommand\textproc{\text}
\usepackage{makecell}
\usepackage{booktabs}
\usepackage{pifont}
\usepackage{multirow}
\usepackage{enumitem}
\usepackage{balance}
\usepackage{threeparttable}
\usepackage{amsmath,amsfonts,mathtools} 
\usepackage{longtable}
\usepackage{amsthm}
\newtheorem{theorem}{Theorem}

\tikzstyle{mybox} = [draw=black, very thick,
    rectangle, rounded corners, inner sep=10pt, inner ysep=13pt]
\tikzstyle{fancytitle} =[fill=black, text=white]

\newcommand{\M}{\textsc{AgenticRag}-R1}

\title{{\textsc{AgenticRag}-\textsc{R1}}: Agentic Reinforcement Learning with Stack Memory for Multi-Step Reasoning, Retrieval and Memorizing}
\author{
  \textbf{Xinke Jiang\textsuperscript{1,2,3,*}},
  \textbf{Yue Fang\textsuperscript{1,2,3,*}},
  \textbf{Zhibang Yang\textsuperscript{1,2,3,*}},
  \textbf{Jiaran Gao\textsuperscript{1,2,3,*}},
  \textbf{Zhixin Zhang\textsuperscript{1,2,3,*}},
  \textbf{Tao Feng\textsuperscript{1}},\\
  \textbf{Rihong Qiu\textsuperscript{1,2,3}},
  \textbf{Wentao Zhang\textsuperscript{1}},
  \textbf{Hongxin Ding\textsuperscript{1,2}},
  \textbf{Ruizhe Zhang\textsuperscript{1,2,3}},
  \textbf{Yongxin Xu\textsuperscript{1,3}},\\
  \textbf{Yuheng Huang\textsuperscript{4}},
  \textbf{Xu Chu\textsuperscript{2,3,5,\ensuremath{\dagger}}},
  \textbf{Junfeng Zhao\textsuperscript{2,3,\ensuremath{\dagger}}},
  \textbf{Yasha Wang\textsuperscript{1,6,\ensuremath{\dagger}}}\\
  \textsuperscript{1}National Engineering Research Center of Software Engineering, Peking University, Beijing, China\\
  \textsuperscript{2}School of Computer Science, Peking University, Beijing, China\\
  \textsuperscript{3}Key Laboratory of High Confidence Software Technologies, Ministry of Education, Beijing, China\\
  \textsuperscript{4}GRG Banking Equipment Co., Ltd., Guangzhou, China\\
  \textsuperscript{5}Center on Frontiers of Computing Studies, Peking University, Beijing, China\\
  \textsuperscript{6}Peking University Information Technology Institute (Tianjin Binhai), Tianjin, China\\ \small
  \textsuperscript{*}Equal contribution \quad
  \textsuperscript{\ensuremath{\dagger}}Corresponding authors\\
  \small{\texttt{\{xinkejiang, yangzb\}@stu.pku.edu.cn;
  \{chu\_xu, zhaojf, wangyasha\}@pku.edu.cn}}
}

\begin{document}
\maketitle
\begin{abstract}
Retrieval-Augmented Generation (RAG) improves the factuality of large language models (LLMs), yet existing RAG systems often struggle with complex, multi-step reasoning that requires \emph{adaptive retrieval and continuous revision of intermediate contexts}.
Recent reinforcement learning (RL)-based agentic RAG methods partially alleviate this issue, but typically rely on \emph{coarse-grained action spaces} and \emph{trajectory-level rewards}, resulting in weak reward assignment and a bias toward short-horizon, stereotyped reasoning template.
To address, we propose \M, a RL framework that deeply integrates reasoning, retrieval, and memory via a memory stack and fine-grained action space, supported by hierarchical action-aware rewards and an information-aware trajectory rejection strategy to enable effective long-horizon learning.
Experiments across a diverse set of multi-hop, open-domain, and agentic reasoning benchmarks, spanning multiple backbone model sizes, demonstrate that \M~consistently outperforms strong baselines.
Moreover, \M~learns more robust, interpretable, and memory-aware reasoning behaviors, highlighting the effect of fine-grained action modeling and information-aware optimization for long-horizon reasoning. 
Our code is anonymous available at 
{\url{https://github.com/jiangxinke/Harness-RL/tree/AgenticRAG-R1-Whitebox}}.
\end{abstract}

\section{Introduction}
\textbf{Large Language Models (LLMs)} have achieved impressive performance across a broad range of natural language processing tasks~\cite{kaplan2020scaling,yang2024qwen2,guo2025deepseek,vu2024gptvoicetasker,chang2024survey}. 
However, despite their strengths, LLMs still remain prone to factual errors~\cite{ji2023survey,cao-etal-2020-factual,10.1145/3571730,HyKGE}, outdated knowledge~\cite{he2022rethinking} and limited domain-specific expertise~\cite{kandpal2023large}, making \textbf{R}etrieval-\textbf{A}ugmented \textbf{G}eneration (RAG)~\cite{GraphRAG,asai2023selfrag,yang2024faima} a promising solution by enabling access to external knowledge for improved accuracy and adaptability. 

\begin{figure}[t]
\centering
\includegraphics[width=\linewidth]{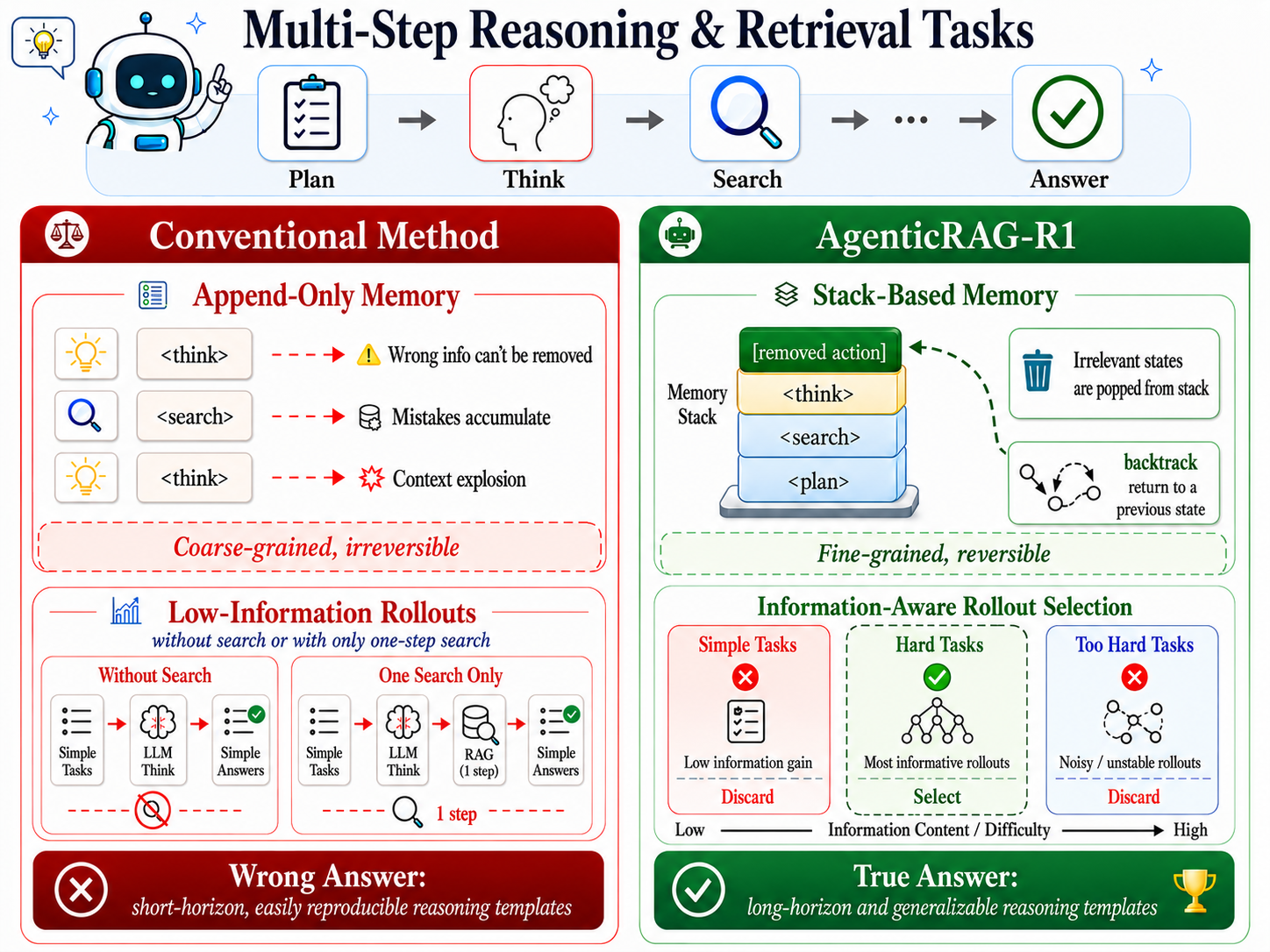}
\caption{Conventional Agentic vs \M~for multi-step reasoning and retrieval.
Conventional methods adopt coarse-grained, append-only memory and shallow retrieval, leading to error accumulation and short-horizon. \M~employs fine-grained, reversible memory actions and information-aware rollout selection, for robust long-horizon reasoning.}
\label{fig:intro}
\end{figure}

In contrast, many real-world decision-making tasks are inherently complicated~\cite{gao2025turnsunlockinglonghorizonagentic}, requiring models to decompose problems into multiple interdependent steps rather than relying on one-off reasoning or single-step retrieval~\cite{shao2023enhancing,press2022measuring}. 
Such settings demand not only access to external knowledge, but also the ability to dynamically decide \emph{when} and \emph{what} to retrieve as the reasoning process unfolds~\cite{lee-etal-2024-planrag,jeong2024adaptive}. 
Crucially, retrieved information could be continuously integrated into the model's evolving internal state, influencing subsequent reasoning and decision making. 
This tight coupling between reasoning and retrieval makes their \textit{\textbf{deep integration}} a critical yet underexplored challenge.

Early attempts to integrate reasoning and retrieval primarily relied on manually crafted prompting strategies~\cite{press2022measuring,shao2023enhancing,trivedi2023interleaving,jiang2024tcrag,li2025searcho1}, and were later extended through \textbf{S}upervised \textbf{F}ine-\textbf{T}uning (SFT) to better coordinate their interaction~\cite{wang2025corag}.
While effective in specific settings, these approaches depend heavily on heuristic design and often overfit to \textbf{fixed reasoning patterns}, without \textbf{intrinsically} improving the model’s ability to jointly reason and retrieve~\cite{chu2025sft}.
More recently, \textbf{R}einforcement \textbf{L}earning (RL) has emerged as a promising paradigm for enhancing complex reasoning, as demonstrated by OpenAI O1~\cite{openaio1} and DeepSeek-R1~\cite{guo2025deepseek}. 
Motivated by this success, several studies have begun to apply RL to retrieval-augmented settings, training models to invoke external retrieval tools during reasoning~\cite{jin2025search,dong2025agentic,dong2025agentic2,zhou2025mem1,liu2026proprietary}. 
As illustrated in Figure \ref{fig:intro}, most existing RL-based RAG rely on \textbf{coarse-grained action modeling}, treating retrieval as a single or shallow action, which limits fine-grained control over complex tasks and biases policy learning toward \textbf{short-horizon}, easily \textbf{reproducible reasoning templates} than long-horizon retrieval and reasoning, posing two challenges:

\noindent\textit{\textbf{\ding{182} Challenge 1. How to expose fine-grained, memory-aware action-making and reward allocation in multi-step reasoning and retrieval.}}
Deep integration of reasoning and retrieval requires LLMs not only to acquire \textbf{external information}, but also to continuously \textbf{manage and revise intermediate content memory} under noisy and uncertain conditions.
In multi-step reasoning, both external retrieval and internal inference inevitably introduce noise: retrieved documents may be irrelevant or misleading~\cite{react,shinn2023reflexion,liu2023lost}, while intermediate reasoning steps may contain logical errors or overconfident assumptions~\cite{jiang2024tcrag,xiong2024watch,hao2023reasoning}.
If such noisy intermediate conclusions are written into context memory without the ability to be inspected, revised, or discarded, they tend to propagate and amplify over steps, effectively corrupting the agent’s internal decision state~\cite{ning2025survey,hoscilowicz2025clickagent,lai2024autowebglm}.
However, most existing RL-based agentic methods adopt a highly \textbf{coarse-grained} action formulation, collapsing the entire reasoning and retrieval processes into a single latent and memory-appending \texttt{<think>} and \texttt{<search>} action template~\cite{dong2025agentic,dong2025agentic2,zhou2025mem1}, which are then stored as irreversible content.
This limitation is further exacerbated by reward design, where outcome rewards are evaluated only on final correctness, overlooking the contributions of intermediate actions throughout the trajectory~\cite{shao2024deepseekmath,jin2025search}.
Without fine-grained, memory-aware action and reward signals, RL suffers from weak reward assignment and is biased toward \textbf{shallow, template-like strategies}.

\noindent\textit{\textbf{\ding{183} Challenge 2. How to expose reinforcement learning to sufficiently diverse and high-information reasoning and retrieval rollouts.}}
Even with fine-grained action spaces and memory-aware reward mechanisms, RL may still fail to learn robust long-horizon reasoning if the training distribution is dominated by \textbf{low information trajectories}. 
Most multi-hop and retrieval-augmented tasks are simple—requiring only a \textbf{single retrieval step or shallow reasoning}—so these short-horizon trajectories often yield higher immediate rewards than more complex, long-horizon ones~\cite{gao2025turnsunlockinglonghorizonagentic,zhou2025mem1}. 
As a result, RL tends to overfit to \textbf{low-information}, trivially rewarding trajectories, rarely exercising actions such as revising, summarizing, or planning, and under-optimizing the long-horizon combinations of cognitive actions~\cite{wang20258020rulehighentropyminority}. 

To address the above challenges, we propose \M, a reinforcement learning framework that achieves deep integration of reasoning, retrieval and memorizing through fine-grained action modeling and action-aligned reward mechanisms.
\ding{182} For \textbf{\textit{C1}}, we design a structured multi-action space together with a stack-based memory which represents the agent's internal state across multi-step reasoning. 
The memory stack allows intermediate actions to be selectively pushed, revised, or popped, enabling the agent to inspect and correct noisy information accumulated during reasoning and retrieval. 
Specifically, our framework supports distinct cognitive actions such as \texttt{<think>}, \texttt{<search>}, and \texttt{<plan>}, as well as memory-oriented actions such as \texttt{<summary>} and \texttt{<backtrack>}, which operate over the memory stack to refine or retract intermediate actions. 
This action design is further complemented by an action-aware reward attribution mechanism that assigns process rewards to each actions, providing fine-grained reward supervision.
\ding{183} For \textbf{\textit{C2}}, we introduce a dynamic training trajectories rejection strategy that rejects low-information trajectories and prioritizes high-reward rollouts while up-weighting inputs whose rollouts exhibit meaningful dispersion. 
Inspired by rejection sampling~\cite{liustatistical,xiong2025minimalist}, we define a \textbf{cross-batch, information-awareness metric (Upper Confidence Bound-like)}~\cite{UCB_1,UCB_2} that evaluates each trajectory based on both its \textit{rollout reward} and the \textit{sample variance}. 
This approach ensures trajectories eliciting diverse reasoning behaviors and requiring multi-step, long-horizon decision-making receive higher priority during training.
In summary, our main contributions are as follows:
\begin{itemize}[leftmargin=*,noitemsep]
\item We first formalize the agentic reinforcement learning process within a search-and-memory environment, and introduce a fine-grained, memory-aware action space together with action-aware process-level rewards. This design enables explicit supervision of heterogeneous cognitive actions and effectively mitigates error accumulation in the agent’s memory (\textit{C1}).
\item We introduce a dynamic, information-aware trajectories rejection strategy that filters low-information trajectories and prioritizes challenging, long-horizon samples with high behavioral variance, ensuring RL effectively explores and optimizes complex, multi-step reasoning strategies (\textit{C2}).
\item We evaluate \M~on multiple multi-hop and RAG benchmarks, demonstrating improved reasoning performance, generalization, and applicability to complex downstream tasks, i.e., report generation.
\end{itemize}




\section{\M}\label{sec:methodology}
In this section, we present \M, as illustrated in Figure~\ref{fig:newmethod.png}, which consists of two key components:
\ding{182} action-centric training paradigm that supports multi-turn, multi-action rollouts with reasoning, retrieval, and memory operations, and incorporates hierarchical, action-aware reward modeling for fine-grained, token-level reward assignment;
\ding{183} Information-Gain-Aware dynamic rollout selection strategy, improves long-horizon exploration and data efficiency during training. More details of Notation Table, Full Algorithm, Prompt, Proofs, are in Appendix~\ref{sec:app notation},~\ref{sec: algorithm},~\ref{sec: prompts},~\ref{sec:proof of agentic}.

\begin{figure*}
    \centering
    \includegraphics[width=1.0\textwidth]{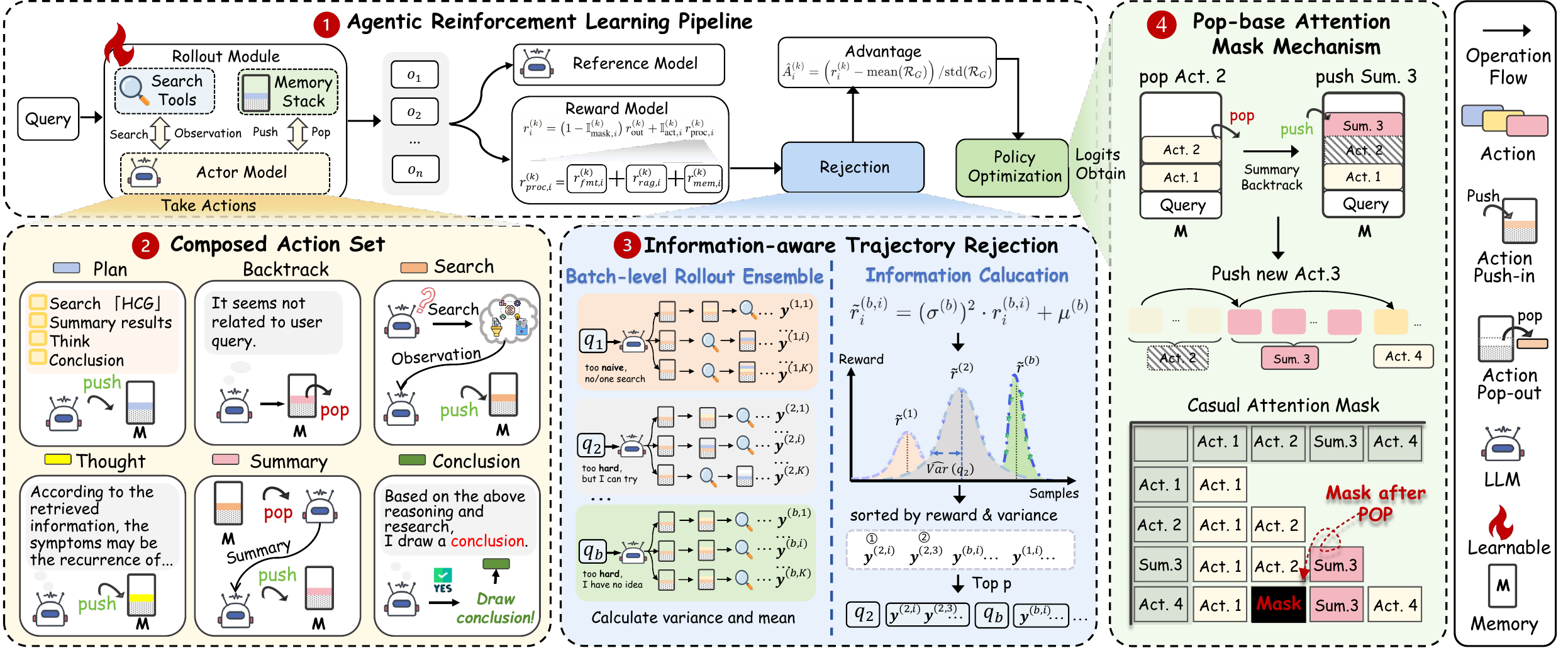}
    \caption{Overall Framework of \M.}
    \label{fig:newmethod.png}
\end{figure*}

\subsection{Action Modeling with Hierarchical Rewards}
\paragraph{Memory-aware Multi-Action Modeling.}
To overcome \textbf{\textit{C1}}, we introduce a structured \emph{multi-action mechanism} integrated with the last-in-first-out (LIFO) memory stack that supports \texttt{push} and \texttt{pop} operations.
Each stack element corresponds to a discrete and semantically meaningful reasoning unit,
serving as the basic execution primitive of the agent.
The memory stack enables step-wise \textbf{action decomposition} of reasoning trajectories,
supports dynamic rollback through stack manipulation,
and maintains a reversible and interpretable record of the agent’s internal decision process:
\begin{itemize}[leftmargin=*, noitemsep]
    \item \textbf{\texttt{<Plan>}} (\texttt{push}): Generate and store a sub-goal, task decomposition, or high-level strategy for long-horizon reasoning.
    \item \textbf{\texttt{<Think>}} (\texttt{push}): Record internal reasoning, such as logical deduction or decision-making, to support subsequent inference.

    \item \textbf{\texttt{<Search>}} (\texttt{push}): Formulate a query with structured parameters and invoke the retrieval module. This action explicitly determines what and how external knowledge is accessed by external tools.    
    \item \textbf{\texttt{<Backtrack>}} (\texttt{pop}, then \texttt{push}): Remove one or more top elements upon detecting inconsistencies or errors to restore a prior valid reasoning state for alternative exploration.
    \item \textbf{\texttt{<Summary>}} (\texttt{pop}, then \texttt{push}): Compress recent memory entries into a concise representation and push the summary back onto the stack, for context management and information preservation.
    \item \textbf{\texttt{<Conclusion>}} (\texttt{push}): Indicate that the current state represents a candidate final answer, concluding the reasoning sequence.

\end{itemize}

\paragraph{Hierarchical Reward Modeling}
To better model interaction outcomes arising from multi-turn reasoning and multi-action execution—particularly those involving \textbf{external retrieval engines and the memory stack}—we introduce a \emph{hierarchical reward modeling} framework. Specifically, we decompose the overall reward signal into two parts: a \emph{outcome reward} and a \emph{process reward}. The outcome reward reflects the quality of the final outcome and is uniformly distributed across valid generation tokens, while the process reward provides fine-grained supervision and is assigned to tokens for corresponding actions.

Formally, for a rollout indexed by $k$, let
\(
y^{(k)} = \bigl(x^{(k)}_1, \dots, x^{(k)}_{|y^{(k)}|}\bigr)
\)
denote the generated token sequence of the $k$-th trajectory.
We define per-token reward as:
\begin{equation}
r^{(k)}_i
=
\bigl(1 - \mathbb{I}^{(k)}_{\text{mask},i}\bigr)\, r^{(k)}_{\text{out}}
+
\mathbb{I}^{(k)}_{\text{act},i}\, r^{(k)}_{\text{proc},i},
\text{ }
i\in[1, |y^{(k)}|],
\end{equation}
where $\mathbb{I}^{(k)}_{\text{mask},i}$ indicates whether token $x^{(k)}_i$ is masked from policy optimization (e.g., retrieved observations) to avoid those false gradient updates, and $\mathbb{I}^{(k)}_{\text{act},i}$ indicates whether the token corresponds to this action invocation. This hierarchical reward formulation allows outcome-level supervision to be shared across the entire reasoning trajectory, while process-level rewards are selectively applied to action tokens, providing fine-grained control over decisions:

\begin{itemize}[leftmargin=*]
    \item \textbf{{Outcome Reward}} $r^{(k)}_{\text{out}}$\textbf{:}  evaluates the quality of the final answer produced by the trajectory. We extract the predicted answer from \texttt{<conclusion>} $a^{(k)}_{\text{conclusion}}$ from $y^{(k)}$ and compare it with the ground-truth $a_{\text{gold}}$ via Exact Match (EM) for discrete answers and F1 score for graded answers:
    \begin{equation}
    r^{(k)}_{\text{out}}
    =
    \mathrm{Acc}\!\left(a^{(k)}_{\text{conclusion}}, a_{\text{gold}}\right)
    \in [0,1].
        \end{equation}

    \item \textbf{Process Rewards} $r^{(k)}_{\text{proc},i}$\textbf{:}
Beyond the final outcome, we introduce action-conditioned process rewards
to supervise intermediate decisions during reasoning.
These rewards are applied only at action tokens
and encourage syntactically valid actions,
effective external retrieval,
and robust memory manipulation.
Specifically, the process reward is defined as:
\begin{equation}
    r^{(k)}_{\text{proc},i}
=
r^{(k)}_{\text{fmt},i}
+
r^{(k)}_{\text{rag},i}
+
r^{(k)}_{\text{mem},i},
\end{equation}
which contains the three types:
\begin{itemize}[leftmargin=*]
    \item 
\textbf{\ding{182} Action Format Reward} ($r^{(k)}_{\text{fmt},i}$),  enforces syntactic correctness at action level.
Let $\mathbb{I}_{\text{syntax}}(a_i)$ be an indicator that equals 1 if action $a_i$ is properly formatted with valid
\texttt{<ACTION>}...\texttt{</ACTION>} tags, and 0 otherwise.
$r^{\text{fmt}}_t$ is defined as
\begin{equation}
    r^{(k)}_{\text{fmt},i} =
\alpha_{\text{type}}(a_i),  \text{  if  } \mathbb{I}_{\text{syntax}}(a_i)=1,
\end{equation}
where $\alpha_{\text{type}}(a_i)=1.0$ for \texttt{CONCLUSION} actions and $0.5$ for other valid action types, otherwise $0$. 
\item \textbf{\ding{183} Search-based Reward} ($r^{(k)}_{\text{rag},i}$). To encourage effective external retrieval, we assign a reward based on semantic relevance of retrieved content. Upon a \texttt{<Search>} action, a large-parameterized \emph{general reward model} (GRM)~\cite{grm} evaluates the alignment between the issued query, the retrieved evidence, and the current task context, producing a relevance score:
\begin{equation}
r^{(k)}_{\text{rag},i}
= \textsc{Grm}\!\left(q^s_i,\, o_i,\, q\right)
\in [0,1],
\end{equation}
where $q^s_i$ denotes the search query, $o_i$ the retrieved observation, and $q$ the original user input.
\item \textbf{\ding{184} Memory-aware Reward}
($r^{(k)}_{\text{mem},i}$) supervises \emph{useful} memory operations rather than
their mere occurrence.
For memory-related actions (\texttt{Backtrack}, \texttt{Summary}),
we also employ the GRM to assess whether the operation is \emph{reasonable}
and \emph{beneficial} given the current reasoning context and memory state.
Specifically, we define:
\begin{equation}
r^{(k)}_{\text{mem},i}
=
\textsc{Grm}\!\left(a_i,\, M_i,\, q\right) \in [0,1],
\end{equation}
where $M_i$ denotes the current memory stack.
We use separate fixed prompts for retrieval and memory evaluation, as detailed in Appendix~\ref{sec: prompts}.
\end{itemize}
\end{itemize}

\noindent\textbf{Reward Budget and Proportional Normalization.}
To prevent reward hacking via excessive action invocation, we impose a per-trajectory budget on process rewards. Let $r^{(k)}_{\text{proc}}=\sum_i r^{(k)}_{\text{proc},i}$ denote the accumulated process reward. Process reward is normalized as:

\begin{equation}
\tilde{r}^{(k)}_{\text{proc},i}
=
\begin{cases}
r^{(k)}_{\text{proc},i}, & R^{(k)}_{\text{proc}} \le B_{\text{proc}}, \\
\dfrac{B_{\text{proc}}}{R^{(k)}_{\text{proc}}}\, r^{(k)}_{\text{proc},i},
& R^{(k)}_{\text{proc}} > B_{\text{proc}},
\end{cases}
\end{equation}
where $B_{\text{proc}}$ is a fixed reward budget.
This ensures that when the total process reward exceeds the budget, all process rewards are proportionally scaled to share the same budget, thereby discouraging excessive or redundant action invocation.

\subsection{Information-Aware Trajectories Rejection}
\label{sec:dynamic-rejection}
Despite fine-grained, action-level reward, \textbf{\textit{C2}} still remains: how to distinguish \emph{informative} trajectories
from noisy or redundant ones when multiple stochastic rollouts are generated per input.
In practice, rollouts sampled from the same prompt can exhibit substantially different reasoning
paths and reward outcomes, while rollouts across different inputs vary widely in their overall
learning utility.
Naively treating all rollouts equally during optimization can dilute gradient signals and bias
training toward \textbf{short-sighted, low-effort, or stereotyped templated behaviors}. To address, we propose an \emph{Information-Aware Trajectories Rejection} mechanism that
explicitly accounts for both \emph{input-level uncertainty} and \emph{rollout-level quality}.
Unlike prior approaches that perform selection independently per input, our method operates
\emph{globally} over a mini-batch of rollouts, enabling competition across both inputs and
trajectories.

\paragraph{Batch-level Rollout Ensemble.}

At each training step, we sample a mini-batch of inputs $\mathcal{B} = \{q^{(1)}, \dots, q^{(B)}\}$. For each input $q^{(b)}$, the policy $\pi_\theta$ generates $K$ stochastic rollouts $y^{(b)} = \{y^{(b,1)}, \dots, y^{(b,K)}\}$, where $y^{(b,i)} \sim \pi_\theta(\cdot \mid q^{(b)})$. All rollouts in the mini-batch form a unified rollout set $\mathcal{Y}_{\mathcal{B}} = \bigcup_{b=1}^{B} y^{(b)}$ with cardinality $|\mathcal{Y}_{\mathcal{B}}| = B \times K$. Each rollout $y^{(b,i)}$ is assigned with a scalar reward $r^{(b,i)}$. For each input $q^{(b)}$, we compute the empirical reward mean and variance across its rollouts to estimate expected performance:
\begin{equation}
\small
\mu^{(b)}=\frac{1}{K} \sum_{i=1}^{K} r^{(b,i)},
\text{ }
(\sigma^{(b)})^2=\frac{1}{K} \sum_{i=1}^{K}
\left(r^{(b,i)} - \mu^{(b)}\right)^2, \notag
\end{equation}

Here, the variance $(\sigma^{(b)})^2$ provides a lightweight measure of \emph{input-level informativeness}, reflecting the variability of learning signals across reasoning and retrieval trajectories for a given input.

\paragraph{Information-Aware Rejection.}
Given the batch-level rollout set
$\mathcal{Y}_{\mathcal{B}} = \{ y^{(b,i)} \}_{b=1,i=1}^{B,K}$,
we assign each rollout an \emph{information-aware score} that jointly
accounts for \emph{input-level diversity} and \emph{rollout-level quality}.
Concretely, for $i$-th rollout of input $x^{(b)}$, we define it as Upper Confidence Bound (UCB) following~\cite{UCB_1, UCB_2}: 
\begin{equation}
\tilde{r}^{(b,i)}
=
(\sigma^{(b)})^2 \cdot r^{(b,i)} + \mu^{(b)},
\end{equation}
where $r^{(b,i)}$ denotes the aggregated scalar reward of the rollout,
and $\mu^{(b)}$ and $(\sigma^{(b)})^2$ are the mean and variance of rollout
rewards under the same input $x^{(b)}$.
This scoring function captures two complementary signals.
The variance term $(\sigma^{(b)})^2$ reflects how much the learning signal
varies across different reasoning and retrieval trajectories induced by the
same input, and thus emphasizes inputs that admit diverse, non-trivial
solution paths.
The rollout-level reward $r^{(b,i)}$ distinguishes higher-quality trajectories
within the same input context.
By multiplicatively coupling these terms, the score selectively amplifies
\textbf{high-reward rollouts when meaningful diversity exists}, while suppressing
uninformative or degenerate trajectories.
Additive mean term $\mu^{(b)}$ serves as a stabilizing baseline to prevent overly rejection.

Rollout rejection is then performed \emph{globally} over the batch-level rollout set.
Specifically, we rank all rollouts $y^{(b,i)} \in \mathcal{Y}_{\mathcal{B}}$
according to $\tilde{r}^{(b,i)}$ and retain only the top fraction
$p \in (0,1]$.
Rollouts not selected are discarded and do not participate in policy optimization.
By performing rejection over the joint \emph{batch $\times$ rollout} space,
this mechanism prioritizes trajectories that are both locally high-quality
and globally informative, encouraging exploration of diverse reasoning
strategies while improving long-horizon reasoning.

\section{Experiments}
We conduct extensive experiments to evaluate \M~and address the following key research questions:

\begin{itemize}[leftmargin=*,noitemsep,topsep=2pt]

    \item \textbf{RQ1:} Does \M~consistently outperform strong baselines on in-domain and out-of-domain QA?
    \item \textbf{RQ2:} How do memory-aware actions and trajectory rejection jointly affect overall model performance?
    \item \textbf{RQ3:} What is the impact of retrieval- and memory-based rewards on the quality of learning and reasoning?
    \item \textbf{RQ4:} How sensitive is \M~to the chosen rejection threshold in trajectory filtering?
    \item \textbf{RQ5:} How does \M~effectively handle long-horizon reasoning with larger step budgets?    
\end{itemize}

\subsection{Experimental Setup}

\textbf{\ding{182} Training Data.}
We train our models and baselines on a curated multi-hop question answering dataset constructed from the training splits of HotpotQA~\cite{yang2018hotpotqa} and 2WikiMultiHopQA~\cite{ho2020constructing}. To focus on genuinely non-trivial reasoning scenarios, we filter out instances that require no external retrieval or can be solved with only a single, trivial retrieval step. Detailed dataset construction procedures and statistics are provided in Appendix~\ref{appendix:datasets}. \textbf{\ding{183} Evaluation Benchmarks.} We evaluate our method on ten benchmarks spanning three settings: \emph{multi-hop QA} (HotpotQA, 2Wiki, MusiQue~\cite{trivedi2022musique}, Bamboogle~\cite{press2022measuring}), \emph{open-domain QA} (Natural Questions~\cite{kwiatkowski2019natural},
TriviaQA~\cite{joshi2017triviaqa}),
and \emph{agentic benchmarks} (FRAMES~\cite{krishna2025fact}).
Additional benchmark details are reported in
Appendix~\ref{appendix:datasets}.
\textbf{\ding{184} Backbone Models.}
We adopt \textbf{Qwen2.5}~\cite{Yang2024Qwen25TR} as the backbone language
model across all experiments, with parameter scales of 1.5B, 3B, and 7B.
\textbf{\ding{185} RAG Tools.}
We employ a Wikipedia-based search tool built on a dump dated November~1,~2023.
\textbf{\ding{186} Baselines.}
We compare our approach against a broad range of baselines spanning
\emph{No-RAG}, \emph{Naive RAG}, \emph{Agentic RAG}, and
\emph{RL-based Agentic} paradigms.
These include Base, CoT~\cite{wei2022chainofthought}, FS-RAG~\cite{trivedi2023interleaving},
FL-RAG~\cite{khandelwal2020generalizationmemorizationnearestneighbor},
ReAct~\cite{react},
IRCoT~\cite{trivedi2023interleaving}, TC-RAG~\cite{jiang2024tcrag},
ReSearch~\cite{chen2025learning}, Search-R1~\cite{jin2025search},
AEPO~\cite{dong2025agentic}, ARPO~\cite{dong2025agentic2},
and Mem1~\cite{zhou2025mem1}. Baseline descriptions are provided in Appendix~\ref{appendix:baselines}. 
\textbf{\ding{187} Evaluation Metrics.} We report the \textbf{F1 score}~\cite{yacouby2020probabilistic} as the primary evaluation metric for all question-answering benchmarks (Appendix~\ref{appendix:Implementation}). 
\ding{188} \textbf{Additional Experiments. } More experiments are listed in Appendix~\ref{sec:additional experiments}.

\subsection{Main Result Analysis}
\label{Main_Result_Analysis}
\begin{table*}[!ht]
\centering
\fontsize{10pt}{12pt}\selectfont
\setlength{\tabcolsep}{3pt}
\renewcommand{\arraystretch}{0.85}
\resizebox{\textwidth}{!}{
\begin{tabular}{l | l | l l | l l l l l | l}
\toprule
\rowcolor{gray!30}
\multicolumn{2}{c|}{\textbf{Method}} &
\multicolumn{2}{c|}{\textbf{In-Domain}} &
\multicolumn{5}{c|}{\textbf{Out-of-Domain}} &
\\

\rowcolor{gray!30}
\textbf{Paradigm} & \textbf{Approach} &
\textbf{2Wiki} & \textbf{HotpotQA} &
\textbf{Bamboogle} & \textbf{FRAMES} &
\textbf{MusiQue} & \textbf{NQ} & \textbf{TriviaQA} &
\textbf{Avg} \\
\midrule

\multicolumn{10}{c}{\textbf{\textit{Qwen2.5-1.5B}}} \\
\midrule

\multirow{2}{*}{No RAG}
&Base & 21.95 & 20.20 & 5.87 & 9.51 & 8.26 & 11.31 & 32.81 & 15.70 \\
& COT  & 16.75 & 17.90 & 19.74 & 8.43 & 7.22 & 9.82  & 25.17 & 15.00 \\
\midrule

\multirow{2}{*}{Naive RAG}
&FS-RAG & 21.37 & 26.55 & 15.31 & 11.40 & 9.50 & 17.74 & 44.73 & 20.94 \\
& FL-RAG & 26.08 & 28.14 & 14.84 & \underline{12.58} & \underline{10.30} & 21.94 & \underline{48.59} & 23.21 \\
\midrule

\multirow{3}{*}{Agentic RAG}
&ReACT & 12.43 & 23.48 & 12.54 & 8.63 & 7.99 & 21.06 & 33.91 & 17.15 \\
& IRCOT & 21.79 & 26.68 & 17.49 & 9.32 & 8.73 & 22.93 & 44.19 & 21.59 \\

&TCRAG & 25.23 & 22.29 & 11.63 & 10.27 & 7.90 & 14.21 & 34.05 & 17.94 \\
\midrule

\multirow{6}{*}{{RL-based}}
& ReSearch  & 25.55 & \underline{30.45} & \textbf{22.16} & 9.21 & 9.54 & \underline{29.07} & 44.65 & \underline{24.38} \\

&Search-R1 & \underline{26.51} & 20.66 & 13.80 & 9.65 & 8.36 & 11.93 & 25.76 & 16.67 \\
& AEPO     & 17.34 & 14.75 & 11.58 & 8.98 & 7.11 & 9.55 & 23.02 & 13.19 \\

&ARPO     & 20.02 & 15.48 & 12.38 & 9.11 & 6.16 & 9.19 & 23.96 & 13.76 \\
& Mem1     & 14.43 & 14.34 & 12.03 & 8.18 & 6.08 & 18.41 & 22.06 & 13.65 \\

&Ours     
& \textbf{29.53} {\color{posgreen}{(+3.02)}} 
& \textbf{30.88} {\color{posgreen}{(+0.43)}} 
& \underline{20.00} {\color{negred}{(-2.16)}} 
& \textbf{13.46} {\color{posgreen}{(+0.88)}} 
& \textbf{16.48} {\color{posgreen}{(+6.18)}} 
& \textbf{31.35} {\color{posgreen}{(+2.28)}} 
& \textbf{51.83} {\color{posgreen}{(+3.24)}} 
& \textbf{27.65} {\color{posgreen}{(+3.27)}} \\
\midrule

\multicolumn{10}{c}{\textbf{\textit{Qwen2.5-3B}}} \\
\midrule

\multirow{2}{*}{No RAG}
&Base & 23.98 & 24.08 & 9.45 & 8.01 & 9.70 & 14.27 & 38.84 & 18.33 \\
& COT  & 18.90 & 23.82 & 20.80 & 7.16 & 10.47 & 16.53 & 39.62 & 19.61 \\
\midrule

\multirow{2}{*}{Naive RAG}
&FS-RAG & 15.47 & 25.85 & 10.48 & 10.42 & 7.64 & 19.84 & 45.38 & 19.30 \\
& FL-RAG & 16.80 & 26.78 & 11.05 & 9.19 & 7.29 & 21.93 & 48.50 & 20.22 \\
\midrule

\multirow{3}{*}{Agentic RAG}
&ReACT & 25.09 & 34.37 & 24.86 & 10.53 & \underline{13.92} & 27.19 & 46.04 & 26.00 \\
& IRCOT & 15.89 & 24.50 & 25.27 & 6.79 & 12.43 & 27.86 & 49.19 & 23.13 \\

&TCRAG & 28.47 & 21.94 & 17.59 & 7.69 & 8.99 & 20.74 & 50.46 & 22.27 \\
\midrule

\multirow{6}{*}{{RL-based}}
& ReSearch  & 27.23 & 33.96 & 15.09 & 10.00 & 9.47 & 34.61 & 53.93 & 26.33 \\

&Search-R1 & \underline{29.90} & \underline{37.24} & \underline{29.90} & 10.76 & 13.53 & \underline{34.73} & \underline{55.08} & \underline{30.16} \\
& AEPO     & 23.01 & 28.71 & 22.09 & 12.26 & 11.70 & 26.76 & 47.78 & 24.62 \\

&ARPO     & 29.55 & 36.48 & 27.32 & \underline{13.49} & 13.38 & 33.29 & 53.66 & 29.60 \\
& Mem1     & 18.06 & 20.15 & 5.19 & 4.99 & 4.47 & 19.18 & 33.09 & 15.02 \\

&Ours
& \textbf{32.92} {\color{posgreen}{(+3.02)}}
& \textbf{44.00} {\color{posgreen}{(+6.76)}}
& \textbf{31.48} {\color{posgreen}{(+1.58)}}
& \textbf{16.23} {\color{posgreen}{(+2.74)}}
& \textbf{16.48} {\color{posgreen}{(+2.56)}}
& \textbf{37.15} {\color{posgreen}{(+2.42)}}
& \textbf{56.62} {\color{posgreen}{(+1.54)}}
& \textbf{33.55} {\color{posgreen}{(+3.39)}} \\
\midrule

\multicolumn{10}{c}{\textbf{\textit{Qwen2.5-7B}}} \\
\midrule

\multirow{2}{*}{No RAG}
&Base & 25.41 & 26.63 & 17.86 & 12.52 & 12.15 & 19.72 & 49.08 & 23.34 \\
& COT  & 23.55 & 29.10 & 37.56 & 17.60 & 14.35 & 22.47 & 49.33 & 27.71 \\
\midrule

\multirow{2}{*}{Naive RAG}
&FS-RAG & 17.71 & 29.21 & 16.86 & 12.52 & 10.74 & 16.82 & 35.02 & 19.84 \\
& FL-RAG & 19.78 & 34.42 & 24.10 & 12.10 & 12.46 & 19.72 & 42.66 & 23.61 \\
\midrule

\multirow{3}{*}{Agentic RAG}
&ReACT & 27.51 & \underline{42.81} & 27.63 & 15.29 & \underline{19.34} & \textbf{30.01} & 54.55 & \underline{31.02} \\
& IRCOT & \underline{36.45} & 26.29 & 21.90 & 6.78 & 8.39 & 19.63 & 49.43 & 24.12 \\

&TCRAG & 29.70 & 40.83 & 25.13 & 16.46 & 17.56 & 29.01 & \underline{54.78} & 30.50 \\
\midrule

\multirow{6}{*}{{RL-based}}
&ReSearch  & 30.03 & 30.39 & 30.42 & 15.61 & 12.58 & 23.69 & 48.25 & 27.28 \\
& Search-R1 & 35.03 & 38.89 & \underline{42.04} & \underline{18.01} & 19.08 & 29.59 & 55.91 & 26.94 \\
& AEPO     & 19.88 & 13.85 & 13.24 & 7.24 & 5.85 & 9.93 & 17.53 & 12.50 \\
& ARPO     & 30.71 & 25.20 & 32.94 & 12.18 & 12.71 & 17.80 & 40.16 & 24.53 \\
& Mem1     & 25.29 & 29.98 & 36.50 & 14.15 & 14.13 & 26.38 & 51.04 & 28.21 \\

&Ours
& \textbf{38.34} {\color{posgreen}{(+1.89)}}
& \textbf{45.15} {\color{posgreen}{(+2.34)}}
& \textbf{49.21} {\color{posgreen}{(+7.17)}}
& \textbf{19.44} {\color{posgreen}{(+1.43)}}
& \textbf{22.01} {\color{posgreen}{(+2.67)}}
& \underline{23.60} {\color{negred}{(-6.41)}}
& \textbf{58.45} {\color{posgreen}{(+2.54)}}
& \textbf{36.60} {\color{posgreen}{(+5.58)}} \\

\bottomrule
\end{tabular}
}
\caption{Performance comparison (\%) on multi-hop, retrieval-augmented benchmarks, including \textit{2Wiki}, \textit{HotpotQA}, \textit{Bamboogle}, \textit{FRAMES}, \textit{MusiQue}, \textit{NQ}, and \textit{TriviaQA}. \textbf{Bold} indicates the best, and \underline{underline} indicates the second.}
\label{tab:comparison}
\end{table*}

\begin{table*}[htbp]
\centering
\fontsize{7pt}{8pt}\selectfont
\setlength{\tabcolsep}{4.6pt}
\renewcommand{\arraystretch}{0.85}

\resizebox{\textwidth}{!}{
\begin{tabular}{l | l | c c | c c c c c | c}
\toprule

\rowcolor{gray!30}
\multicolumn{2}{c|}{\textbf{Method}} & 
\multicolumn{2}{c|}{\textbf{In-Domain}} & 
\multicolumn{5}{c|}{\textbf{Out-of-Domain}} &
\\

\rowcolor{gray!30}
\multicolumn{1}{c|}{\textbf{Paradigm}} & \multicolumn{1}{c|}{\textbf{Approach}} &
\textbf{2Wiki} & \textbf{HotpotQA} &
\textbf{Bamboogle} & \textbf{FRAMES} &
\textbf{MusiQue} & \textbf{NQ} & \textbf{TriviaQA} &
\textbf{Avg} \\

\midrule

\multirow{1}{*}{Full Model}
& \M    & 32.92 & 44.00 & 31.48 & 16.23 & 16.48 & 37.15 & 56.62 & 33.55 \\

\midrule

\multirow{3}{*}{{\makecell{Component\\Ablation}}}
& w/o MemAction & 27.54 & 28.77 & 24.60 & 9.83 & 13.86 & 35.10 & 51.91 & 27.37 \\
& w/o Reject  & 26.86 & 28.33 & 27.69 & 12.98 & 12.68 & 23.67 & 45.80 & 25.43 \\
& w/o MemAct\&Reject  & 24.56 & 28.61 & 25.49 & 11.61 & 12.54 & 22.07 & 44.83 & 24.24 \\
\midrule

\multirow{3}{*}{{\makecell{Reward\\Ablation}}}
& w/o $r_{\text{rag}}$ & 23.70 & 34.71 & 19.68 & 7.70 & 11.45 & 29.07 & 52.39 & 25.53 \\
& w/o $r_{\text{mem}}$ & 27.49 & 33.78 & 17.14 & 7.41 & 7.63 & 22.00 & 55.29 & 24.39 \\
& w/o $r_{\text{rag}} \& r_{\text{mem}}$  & 22.98 & 26.42 & 15.74 & 5.33 & 6.07 & 18.62 & 44.40 & 19.94 \\

\bottomrule
\end{tabular}
}
\caption{Ablation analysis of component and reward designs in \M~on Qwen2.5-3B.}
\label{tab:ablation}
\end{table*}

To address \textbf{RQ1}, we evaluate \M~on two types of tasks using three backbones (Table~\ref{tab:comparison}) and analyze its training behaviors and learning dynamics (Figure~\ref{fig:reward_behavior_count}). 
Based on these analyses, we highlight the following main observations:

\textbf{\ding{182} Strong Performance Compared with Baselines.}
Across in-domain and out-of-domain multi-hop QA benchmarks (Table~\ref{tab:comparison}),
\M~achieves the best or near-best performance across backbone model sizes.
On the in-domain datasets, \M~consistently surpasses the strongest baselines,
with, for example, absolute gains of \textbf{+3.02\%} on \textit{2Wiki}
and \textbf{+6.76\%} on \textit{HotpotQA} using the 3B backbone.
On out-of-domain benchmarks, \M~exhibits strong generalization,
e.g., achieving an absolute improvement of about \textbf{+7.2\%} on \textit{TriviaQA}
with the 1.5B backbone and \textbf{+7.1\%} on \textit{Bamboogle} with the 7B backbone
over the strongest Agentic RAG baselines.
Notably, these gains hold across backbone scales, demonstrating the effectiveness and robustness of our reinforcement learning-based optimization in this setting.

\textbf{\ding{183} Learning Dynamics of Cognitive Actions and Information-Aware Rejection.}
We examine the evolution of cognitive actions during training and the associated learning dynamics, corroborating trends using reward trajectories from Figure~\ref{fig:wandb_all}. Using the Qwen2.5-3B model as a primary example, \texttt{<think>} actions steadily increase, reflecting a growing reliance on internal reasoning. Actions related to intermediate reasoning, namely \texttt{<summary>} and \texttt{<backtrack>}, fluctuate but overall increase; these behaviors help manage memory and remove noise from previous contexts. In contrast, \texttt{<search>} remains relatively stable with minor increases. \texttt{<conclusion>} actions rise consistently, reducing cases where the model fails to produce a final answer.
For the smaller Qwen2.5-1.5B model, similar trends are observed for \texttt{<think>}, \texttt{<summary>}, \texttt{<backtrack>}, and \texttt{<plan>}. However, \texttt{<search>} frequency is lower, likely due to limited context capacity: performing excessive retrievals can disrupt reasoning, so the model reduces search while still successfully completing reasoning episodes, as reflected by increased \texttt{<conclusion>}.

Additional training statistics track overall model behavior throughout training (Figure~\ref{fig:wandb_all}). Over 0--300 training steps, rewards steadily increase while action entropy shows a general downward trend, suggesting convergence toward more consistent policies. Response length and number of reasoning turns initially rise and then decline, indicating that the model discovers more efficient reasoning and retrieval strategies over time, ultimately producing answers faster and more accurately. Collectively, these analyses reveal how \M~adapts its cognitive actions and information-aware rejection throughout training, balancing exploration and exploitation to optimize multi-hop reasoning.


\begin{figure*}[t]
\begin{minipage}[t]{0.45\textwidth}
    \centering
    \includegraphics[width=\linewidth]{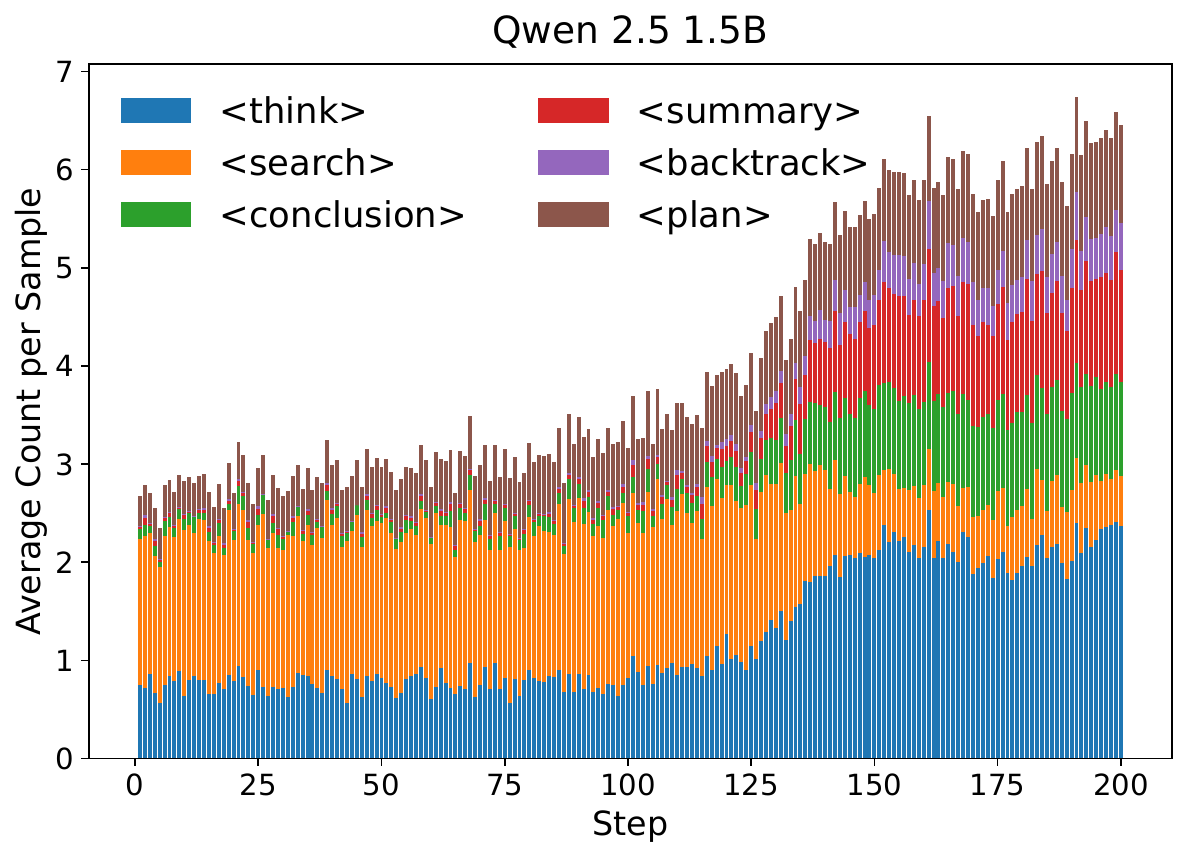}
\end{minipage}
\hfill
\begin{minipage}[t]{0.45\textwidth}
    \centering    \includegraphics[width=\linewidth]{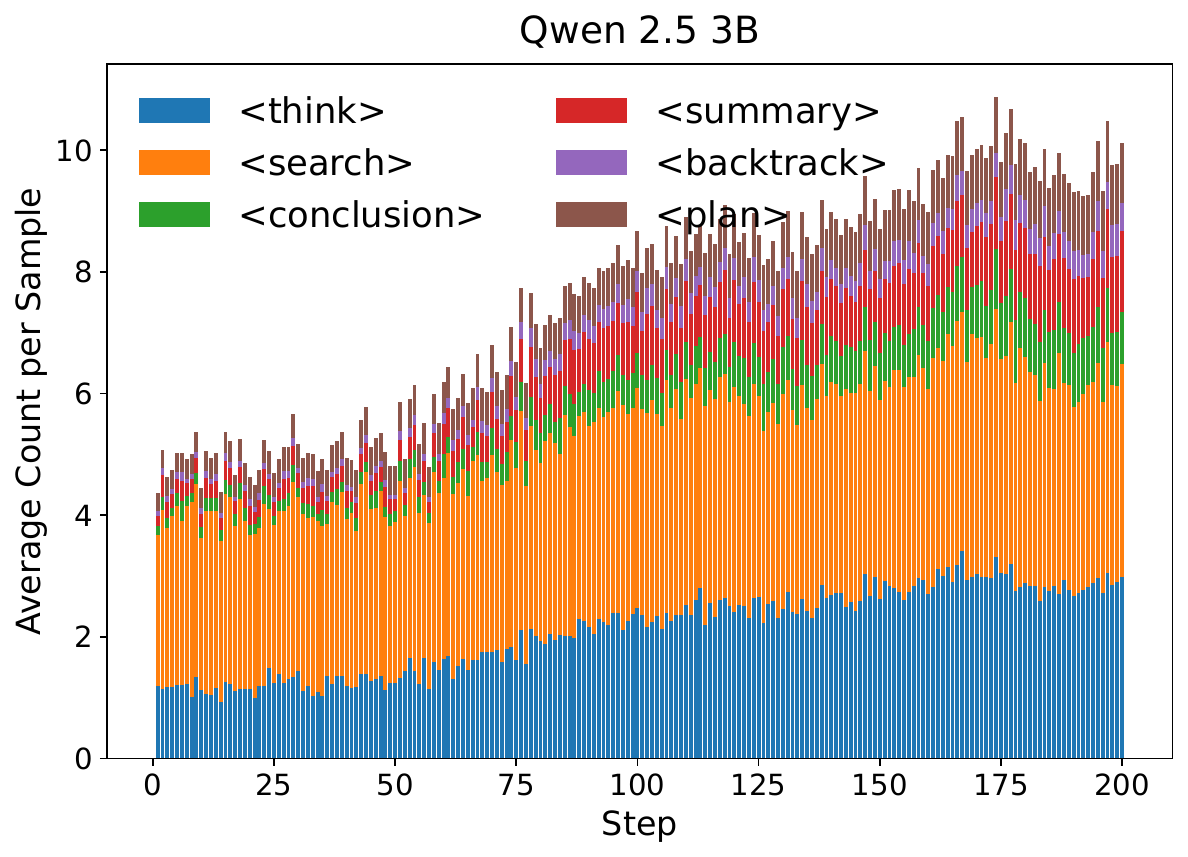}
\end{minipage}
\caption{Action-level frequency distributions during RL training across training steps.}
\label{fig:reward_behavior_count}
\end{figure*}

\subsection{Component Analysis}
We analyze the effects of key model components, process-level rewards, and rejection thresholds on performance across both in-domain and out-of-domain tasks.

\textbf{\ding{182} Model Component Ablation.}
We evaluate key components of \M~in Table~\ref{tab:ablation} for \textbf{RQ2}. Removing memory-related actions (\textit{w/o MemAction}) substantially degrades performance, e.g., 2Wiki drops from 32.92$\%$ to 27.54$\%$ and HotpotQA from 44.00$\%$ to 28.77$\%$. Removing information-aware rejection (\textit{w/o Reject}) also reduces the average score from 33.55$\%$ to 25.43$\%$, while removing both components further lowers it to 24.24$\%$, confirming their complementary contributions. 

\textbf{\ding{183} Process-level Reward Ablation.}
Table~\ref{tab:ablation} also analyzes reward design for \textbf{RQ3}. Removing the RAG reward (\textit{w/o $r_{\text{rag}}$}) hurts 2Wiki (32.92$\%$ to 23.70$\%$) and HotpotQA (44.00$\%$ to 34.71$\%$), while removing the memory reward (\textit{w/o $r_{\text{mem}}$}) weakens out-of-domain performance, e.g., MusiQue drops from 16.48$\%$ to 7.63$\%$. Removing both rewards yields the largest drop (Avg 19.94$\%$), showing that both signals are necessary for  effective reasoning and memory utilization.

\begin{figure}[t]
\centering
\includegraphics[width=0.9\linewidth]{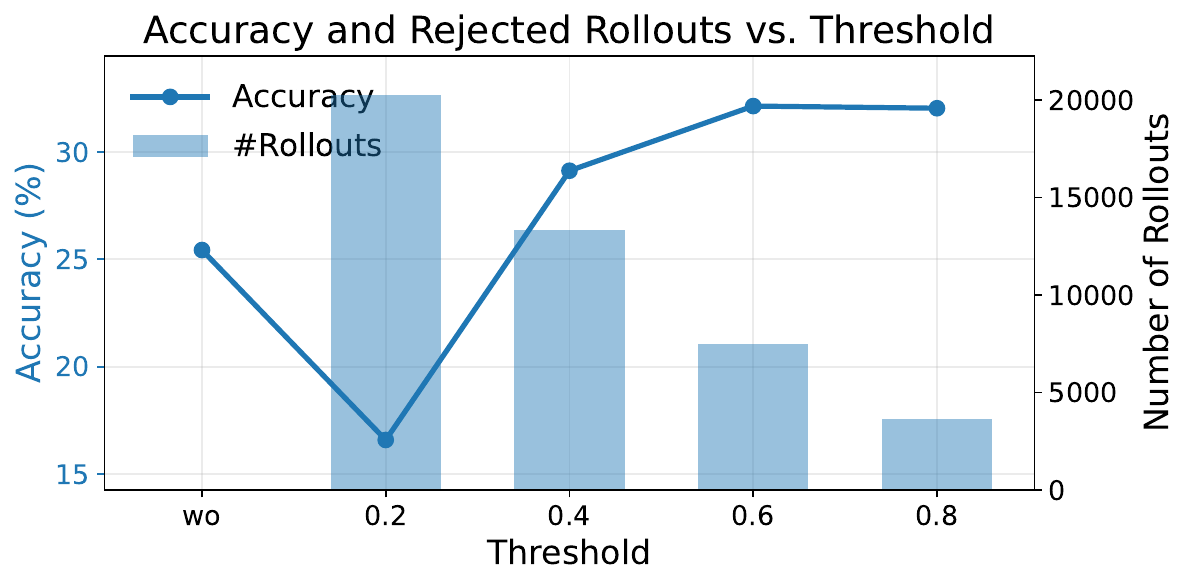}
\caption{Sensitivity analysis of Rejection threshold on Qwen2.5-3B. The line plot shows F1-Score, and the bar chart depicts the number of rejected training samples.}
\label{fig:ole_threshold_analysis}
\end{figure}

{\ding{184} \textbf{Sensitivity to Rejection Threshold.}}
We analyze the impact of varying the rejection threshold on model performance and training behavior (Figure~\ref{fig:ole_threshold_analysis})  to address \textbf{RQ4}. Accuracy, measured by F1-Score, initially decreases and reaches a minimum around a threshold of 0.2, then steadily improves as the threshold increases, peaking at 0.6 and 0.8. Correspondingly, the number of rejected training samples decreases with higher thresholds, from over 20,000 at 0.2 to the fewest at 0.8. This trend indicates that moderate to high thresholds effectively balance answer generation and rejection.

\paragraph{\ding{185}Long-Horizon Experiments.}
To address \textbf{RQ5}, we progressively increase the maximum number of allowed steps from 10 to 15, 20, 25, and 30, thereby evaluating its capability to perform long-horizon reasoning effectively.



\begin{table}[ht]
\centering
\fontsize{9pt}{9.5pt}\selectfont
\renewcommand{\arraystretch}{1.1}
\setlength{\tabcolsep}{5pt}

\resizebox{0.9\linewidth}{!}{
\begin{tabular}{l| ccccc}
\toprule
\rowcolor{gray!30}
\textbf{Method} 
& \multicolumn{5}{c}{\textbf{Max Steps}} \\
\rowcolor{gray!30}
& \textbf{10} & \textbf{15} & \textbf{20} & \textbf{25} & \textbf{30} \\
\midrule
\rowcolor{gray!10}
\multicolumn{6}{c}{\textit{2Wiki}} \\
\midrule
ReAct               & 18.23 & 22.45 & 25.17 & 27.89 & 29.03 \\
Search-R1            & 24.12 & 38.02 & 37.51 & 33.08 & 44.27 \\
TC-RAG              & 22.78 & 37.35 & 36.92 & 32.41 & 43.10 \\
\M                  & 23.91 & 31.81 & 32.53 & 38.61 & 47.45 \\
\midrule
\rowcolor{gray!10}
\multicolumn{6}{c}{\textit{Bamboogle}} \\
\midrule
ReAct               & 15.67 & 18.92 & 20.34 & 21.76 & 22.19 \\
Search-R1            & 37.05 & 34.11 & 37.02 & 28.03 & 28.19 \\
TC-RAG              & 36.42 & 33.34 & 36.36 & 27.27 & 27.23 \\
\M                  & 10.03 & 31.81 & 14.29 & 21.43 & 35.47 \\
\midrule
\rowcolor{gray!10}
\multicolumn{6}{c}{\textit{TriviaQA}} \\
\midrule
ReAct               & 28.76 & 30.12 & 31.05 & 31.42 & 31.68 \\
Search-R1            & 34.21 & 36.12 & 33.57 & 33.09 & 32.28 \\
TC-RAG              & 33.57 & 35.48 & 32.91 & 32.44 & 31.62 \\
\M                  & 37.53 & 32.53 & 39.45 & 43.50 & 55.28 \\

\bottomrule
\end{tabular}
}
\caption{Long-horizon performance under different maximum step budgets across multiple datasets.}
\label{tab:horizon_comparison}
\end{table}
As shown in Table~\ref{tab:horizon_comparison}, our model exhibits clear and consistent gains as the maximum step budget increases, demonstrating strong scalability in long-horizon settings. This trend holds across all evaluated datasets, indicating that the model can effectively leverage extended interaction horizons to perform deeper multi-step reasoning and retrieval. In contrast, \textsc{TC-RAG} shows unstable or even degraded performance when the horizon grows, suggesting limited ability to accumulate and utilize information over long trajectories. These results highlight that simply increasing the allowed steps is insufficient; rather, the proposed design is crucial for unlocking effective long-horizon reasoning capabilities. Furthermore, compared to ReAct and Search-R1, which either plateau early or fluctuate inconsistently across step budgets, \M~exhibits monotonically improving trends on most benchmarks, further validating its effectiveness in long-horizon settings.

\textbf{\ding{186} Individual Action Ablation.}
Table~\ref{tab:action_level_ablation} analyzes the contribution of individual actions. Removing \texttt{<Plan>} decreases the average F1 score by 1.95\%, while removing \texttt{<Summary>} and \texttt{<Backtrack>} causes larger drops of 3.05\% and 3.85\%, respectively. Removing both memory actions yields the largest degradation (6.18\%), demonstrating that these actions are complementary for error correction and context compression.


\begin{table}[t]
\centering
\fontsize{8pt}{9pt}\selectfont
\setlength{\tabcolsep}{4.6pt}
\renewcommand{\arraystretch}{0.9}

\begin{tabular}{lccc}
\toprule
\rowcolor{gray!30}
\textbf{Variant} & \textbf{2Wiki} & \textbf{HotpotQA} & \textbf{Avg.} \\
\midrule

Full \M & \textbf{32.92} & \textbf{44.00} & \textbf{33.55} \\

\midrule

w/o \texttt{<Plan>} & 31.00 & 41.80 & 31.60 \\
w/o \texttt{<Summary>} & 29.90 & 40.10 & 30.50 \\
w/o \texttt{<Backtrack>} & 28.80 & 38.60 & 29.70 \\
w/o both memory actions  & 27.54 & 28.77 & 27.37 \\

\bottomrule
\end{tabular}

\caption{Leave-one-action-out ablation of \M~on Qwen2.5-3B.}
\label{tab:action_level_ablation}
\end{table}

\textbf{\ding{187} Effect of SFT Initialization.}
We compare SFT only, cold-start RL, and RL with SFT initialization on Qwen2.5-3B. Their average F1 scores across seven benchmarks are $22.40 \pm 0.50$$\%$, $33.55 \pm 0.42$$\%$, and $35.20 \pm 0.40$$\%$, respectively. SFT alone underperforms RL, while SFT initialization further improves RL performance. These results indicate that SFT is not required for our framework, but stronger initialization remains beneficial.

\section{Related Works}
\paragraph{\textbf{Reinforcement Learning for LLMs}}
Reinforcement learning has recently shown strong effectiveness in improving LLM reasoning and tool-use capabilities, as demonstrated by OpenAI-O1~\cite{openaio1} and DeepSeek-R1~\cite{guo2025deepseek}. 
While policy gradient methods such as PPO~\cite{schulman2017proximal} are widely used, they incur high computational costs due to repeated on-policy updates. 
Alternative approaches, including DPO~\cite{rafailov2023direct} and GRPO~\cite{shao2024deepseekmath}, improve training efficiency but face challenges such as off-policy bias or limited exploration.
Recent advances target multi-step reasoning and tool interaction, proposing techniques to stabilize training and improve exploration~\cite{dong2025agentic,dong2025agentic2,zhou2025mem1}. 
Nevertheless, existing RL methods remain limited by coarse actions, sparse rewards, and error accumulation in long-horizon reasoning.

\section{Conclusion and Future Work}
We presented \M, a reinforcement learning framework for agentic retrieval-augmented reasoning that integrates fine-grained agentic actions, structured memory control, and action-aware rewards. Experiments across multiple reasoning benchmarks show that \M~consistently outperforms strong RAG and agentic baselines while supporting more effective long-horizon reasoning.
Several challenges remain for future work. In particular, jointly optimizing high-level actions and their associated arguments, as well as resolving conflicts among different agentic actions, remains difficult and calls for more principled coordination and credit assignment. We also plan to evaluate \M~in broader vertical domains and more open-ended scenarios, such as deep research and long-horizon tasks.

\section*{Acknowledgments}
This work was supported by the National Natural Science Foundation of
  China (Grant No.~62506010).

\section*{Limitations}
Despite the promising results, our framework does have some limitations that need to be addressed.  First, the computational overhead introduced by the memory stack system and dynamic retrieval processes can be relatively high.  While these components are essential for managing complex queries and minimizing error accumulation, they do contribute to increased computational costs.  In large-scale or real-time applications, this could present challenges related to processing speed and efficiency.  Although we have proposed solutions like speculative sampling to alleviate some of this overhead, further optimization is still needed to enhance performance and scalability in time-sensitive environments.
Second, our training strategy relies on a rejection-based sampling mechanism that prioritizes more challenging examples to enhance reasoning ability.  While effective, this approach requires careful tuning of several hyperparameters, as different configurations can lead to varying performance outcomes.  This additional tuning effort may complicate practical deployment and affect performance when transferring the framework to new datasets or tasks.  Developing more adaptive or self-calibrating sampling strategies to reduce reliance on manual hyperparameter selection is therefore an important direction for future work.

\section*{Ethical considerations}
To evaluate the performance of the proposed method, we carried out experiments exclusively on publicly accessible benchmark datasets, such as GAIA, HotpotQA, MusiQue, and other widely used resources, in accordance with their respective licenses and usage policies. 
We further confirm that no personally identifiable information was used in this study, and that no human or animal subjects were involved in the research.




\bibliography{custom}

\appendix

\newpage

\addtocontents{toc}{\protect\setcounter{tocdepth}{2}}

\section*{Appendix}
\tableofcontents

\section{More Related Work}
\subsection{Reinforcement Learning for LLMs}
Reinforcement learning (RL) has emerged as a powerful paradigm for improving LLM alignment and reasoning, with recent work such as OpenAI-O1~\cite{openaio1} and DeepSeek-R1~\cite{guo2025deepseek} demonstrating strong performance on complex tasks~\cite{ye2026rubric,zhou2026look,lv2026pcsd,liu2025auditing,xie2026unlocking,xie2026edgeexperiencedistillationguidedexploration}. 
Methods like PPO~\cite{schulman2017proximal} have been widely adopted, but require multiple rounds of on-policy optimization, making them computationally expensive and challenging to scale. To improve efficiency, alternatives such as Direct Preference Optimization (DPO)~\cite{rafailov2023direct} bypass reward modeling by directly optimizing preference loss, but suffer from off-policy bias. Group Relative Policy Optimization (GRPO)~\cite{shao2024deepseekmath} offers a compromise by avoiding learned value networks and leveraging intra-group comparisons for relative rewards.
Building on these, recent RL methods specifically designed for multi-step reasoning and tool interaction include AEPO~\cite{dong2025agentic}, which addresses entropy imbalance and rollout collapse in multi-round agents; ARPO~\cite{dong2025agentic2}, which guides exploration of uncertain high-entropy tool calls using advantage attribution; and MEM1~\cite{zhou2025mem1}, which synergizes memory and reasoning to support long-horizon multi-step interactions. Despite these advances, existing RL approaches still face key challenges: \textbf{coarse-grained action modeling, trajectory-level scalar rewards, limited exploration, cumulative errors in multi-step reasoning, and the sparsity of training samples that require diverse reasoning behaviors}.

\section{Training Algorithm}
\label{sec: algorithm} 
Algorithm \ref{alg:llm_rollout} details the full rollout engine with
a maximum action budget~$B$.

\begin{algorithm}[H]
\footnotesize
\caption{Dynamic Rollout with Multi-Action Reasoning and Stack-Based Memory}
\label{alg:llm_rollout}
\begin{algorithmic}[1]
\Require query $q$, policy $\pi_\theta$, retriever $\mathcal{R}$, budget $T$.
\Ensure Generated reasoning trajectory $y$

\State Initialize trajectory $y \leftarrow \emptyset$
and memory stack $\mathcal{M} \leftarrow \emptyset$

\While{$t < T$}
    \State Initialize action segment token $x_t \leftarrow \emptyset$
    \While{$x$ is not an action delimiter}
        \State Sample token $x \sim \pi_\theta(\cdot \mid q, y, \mathcal{M})$
        \State Append $x$ to $x_t$
    \EndWhile

    \State Append $x_t$ to trajectory $y$
    \State Parse composed action $a_t \leftarrow \mathrm{parse}(x_t)$

    \If{$a_t \in \{\texttt{Plan}, \texttt{Think}, \texttt{Search}, \texttt{Conclusion}\}$}
        \State \textbf{push}($\mathcal{M}$, content of $x_t$)
    \EndIf

    \If{$a_t == \texttt{Search}$}
        \State Extract query $q^s_t \leftarrow \mathrm{parse}(x_t)$
        \State Retrieve observation $o_t \leftarrow \mathcal{R}(q_t)$
        \State Append $\langle\texttt{Observation}\rangle o_t \langle/\texttt{Observation}\rangle$ to $y$
        \State \textbf{push}($\mathcal{M}$,$\langle\texttt{Observation}\rangle o_t \langle/\texttt{Observation}\rangle$)
        \State Modify Retrieval Loss Mask

    \ElsIf{$a_t == \texttt{Backtrack}$}
        \State \textbf{pop}($\mathcal{M}$),  \textbf{push}($\mathcal{M}$, revised reasoning from $x_t$)
        \State Modify Pop-based Attention Mask
    \ElsIf{$a_t == \texttt{Summary}$}
        \State \textbf{pop}($\mathcal{M}$),  \textbf{push}($\mathcal{M}$, summarized content)
        \State Modify Pop-based Attention Mask
    \ElsIf{$a_t == \texttt{Conclusion}$}
        \State \textbf{Return} $y$
    \EndIf

    \State $t \leftarrow t + 1$
\EndWhile

\State \textbf{Return} $y$
\end{algorithmic}
\end{algorithm}

\section{Notations Table}
\label{sec:app notation}
Table~\ref{tab:notations} presents a comprehensive list of key notations and symbols used in the \M\ framework.

\begin{table}[htbp]
\centering
\small
\begin{tabular}{p{1cm} p{6.2cm}} 
\toprule
\rowcolor{gray!10}
\textbf{Symbol} & \textbf{Description} \\
\midrule

$\mathcal{S}$ & System state capturing the current reasoning context \\
$\mathcal{A}$ & Set of reasoning actions, consisting of primitive actions (\texttt{push}/\texttt{pop}) and composed actions (e.g., \texttt{Plan}, \texttt{Search}) \\
$\mathcal{M}$ & Stack-based short-term memory of reasoning units \\
$\delta$ & State transition function over system state and memory \\

$s_t$ & System state at reasoning step $t$ \\
$a_t$ & Reasoning action taken at step $t$ \\
$\mathcal{M}_t$ & Memory stack at step $t$ \\

\hline

$q$ & Input query sampled from the training dataset $\mathcal{D}$ \\
$y$ & Generated reasoning trajectory $(a_1,\dots,a_T)$\\
$\pi_\theta$ & Policy model parameterized by $\theta$ \\
$\pi_{\mathrm{ref}}$ & Frozen reference policy for KL regularization \\
$K$ & Number of rollouts generated per input \\
$\mathcal{R}$ & External retrieval module or search engine \\
$r_\phi(q,y)$ & Scalar reward function evaluating trajectory $y$ given input $q$ \\
$x^{(k)}_{i}$ & The $i$-th generated token in the $k$-th trajectory $y^{(k)}$. \\
$R_G$ & The set of token-level rewards across all trajectories in the current rollout group \\
$r^{(k)}_{i}$ & The token-level reward assigned to the token $x^{(k)}_{i}$. \\
$\hat{A}_{i,t}$ & Group-relative advantage for trajectory $i$ at step $t$ \\
$\mathcal{J}(\theta)$ & The Group Relative Policy Optimization (GRPO) objective. \\

$\tilde{z}_{i,t}$ & Clipped importance sampling ratio for trajectory $i$ at step $t$ \\
$\mathbb{D}_{\mathrm{KL}}(\cdot\|\cdot)$ & Kullback--Leibler divergence between policies \\
$\beta$ & KL regularization coefficient \\

\hline

$r_{\text{out}}^{(k)}$ & Outcome-level reward evaluating final answer \\
$r_{\text{proc}}^{(k)}$ & Process-level reward in the $k$-th trajectory\\
$\mathbb{I}^{(k)}_{\text{mask},i}$ & Indicator function that equals 1 if the token $x_k^{(i)}$ is masked from policy optimization, and 0 otherwise. \\
$\mathbb{I}^{(k)}_{\text{act},i}$ & Indicator function that equals 1 if the token $x_k^{(i)}$ corresponds to an action invocation, and 0 otherwise. \\
$a^{(k)}_{\text{conclusion}}$ & The predicted final answer extracted from the \texttt{<Conclusion>} action in the $k$-th trajectory $y^{(k)}$ \\
$a_{\text{gold}}$ & The ground-truth answer used for evaluation \\
$r^{(k)}_{\text{fmt},i}$ & Action format reward enforcing syntactic correctness at action level\\
$r^{(k)}_{\text{rag},i}$ & Search-based reward evaluating retrieval relevance \\
$r^{(k)}_{\text{mem},i}$ & Memory-aware reward for backtracking or summarization \\
$\mathbb{I}_{\text{syntax}}(a_i)$ & Indicator for whether $a_i$ has valid <ACTION> tags (1 = valid, 0 = invalid). \\
$\alpha_{\text{type}}(a_i)$ & Action weight: 1.0 for CONCLUSION, 0.5 for other valid actions, 0 if invalid. \\

$q^s_i$ & Search query issued at step $i$ \\
$o_i$ & Retrieved observation at step $i$ \\

\hline

$\mathcal{B}$ & Mini-batch size (number of input prompts) \\
$\mathcal{Y}_{\mathcal{B}}$ & Set of all rollouts in a batch \\
$\mathcal{Y}_{\mathcal{B}}^{\mathrm{keep}}$ & Rollout subset retained after global selection \\
$\mu^{(b)}$ & Mean reward over rollouts for input $x^{(b)}$ \\
$(\sigma^{(b)})^2$ & Reward variance measuring input-level uncertainty \\
$\tilde{r}^{(b)}_i$ & Information-aware score for rollout $\tau^{(b)}_i$ \\
$p$ & Fraction of top-ranked rollouts retained\\

\bottomrule
\end{tabular}
\caption{Key Notations Used in \M}
\label{tab:notations}
\end{table}

\section{Experiment Datasets}
\label{appendix:datasets}
\subsection{Training Dataset.}
A central challenge in RL-based agent training lies in the limited difficulty of existing QA datasets.
Effective supervision requires questions \textbf{that demand non-trivial reasoning and long-horizon, multi-step retrieval}.
However, widely used multi-hop benchmarks such as HotpotQA and 2WikiMultiHopQA can often be solved with shallow reasoning and zero or only one retrieval step, making them insufficient for learning complex search and planning behaviors.
As a result, models trained on such data \textbf{tend to overfit short-horizon, template-like reasoning patterns and generalize poorly to challenging agentic scenarios}.

To address this limitation, we construct a curated training dataset following~\cite{gao2025turnsunlockinglonghorizonagentic} by filtering instances from HotpotQA and 2Wiki that are both solvable and intrinsically require multi-step search.
Specifically, for each candidate question, we prompt the training-free TC-RAG~\cite{jiang2024tcrag} to generate 10 independent trajectories under identical conditions, and discard questions that
\ding{182} are answered incorrectly in all trials,
\ding{183} achieve a correctness rate of 50\% or higher, or
\ding{184} can be answered correctly with at most a single retrieval step.

\subsection{Testing set.}

In our experiments, we evaluate our method on 9 widely used datasets spanning multi-hop reasoning, open-domain question answering, and real-world agent evaluation, including HotpotQA~\cite{yang2018hotpotqa}, 2WikiMultiHopQA~\cite{ho2020constructing}, Bamboogle~\cite{press2022measuring}, FRAMES~\cite{krishna2025fact}, 
MusiQue~\cite{trivedi2022musique}, NQ~\cite{kwiatkowski2019natural}
and  TriviaQA~\cite{joshi2017triviaqa}) 
. Dataset statistics, including the number of training, development, and test instances, are provided in Table~\ref{tab:dataset_overview}:
\begin{itemize}[leftmargin=*]
\item \textbf{2WikiMultiHopQA}~\cite{ho2020constructing} is a large-scale multi-hop question answering dataset constructed from Wikipedia and Wikidata. It contains 192,606 question--answer pairs grounded in a single database, with 154,878 training examples, 12,576 development examples, and 12,576 test examples. The dataset is designed to evaluate multi-hop reasoning that requires aggregating evidence across multiple documents.
\item \textbf{HotpotQA}~\cite{yang2018hotpotqa} is a large-scale dataset with 112,779 Wikipedia-based question-answer pairs requiring multi-hop reasoning over multiple documents. The dataset is divided into several subsets for training and evaluation. Specifically, the training set includes 18,089 single-hop questions (\textit{train-easy}), 56,814 multi-hop questions (\textit{train-medium}), and 15,661 hard multi-hop questions (\textit{train-hard}). The development set contains 7,405 hard multi-hop questions, while the test set is further divided into two settings: \textit{test-distractor} with 7,405 hard multi-hop questions and \textit{test-fullwiki} with another 7,405 hard multi-hop questions. This division allows for a comprehensive evaluation of models under different conditions. HotpotQA also provides sentence-level supporting facts required for reasoning, enabling models to improve performance and make explainable predictions.
\item \textbf{Bamboogle}~\cite{press2022measuring} is a small, handcrafted dataset with 125 questions designed to measure the extent to which a question answering system can answer varied compositional questions. It covers many different types of questions on various areas.
\item \textbf{FRAMES}~\cite{krishna2025fact} is a high-quality dataset designed to test 
LLMs'
factual responses, retrieval capabilities, and reasoning in generating final answers. It comprises 824 challenging multi-hop questions requiring integration of information from multiple sources.
\item \textbf{MuSiQue}~\cite{trivedi2022musique} is a multihop question answering dataset constructed from Wikipedia, focusing on questions that require connected reasoning across multiple steps. It contains 19,938 training examples, 2,417 development examples, and 2,459 test examples. MuSiQue is designed to enforce proper multihop reasoning by ensuring that each reasoning step critically relies on the output of the previous step.
\item \textbf{Natural Questions (NQ)}~\cite{kwiatkowski2019natural} is a large-scale question answering dataset based on real user queries issued. 
Each example pairs a question with a Wikipedia page and includes annotated long and short answers. The dataset contains 307,373 training examples, 7,830 development examples, and 7,842 test examples.
\item \textbf{TriviaQA}~\cite{joshi2017triviaqa} is a large-scale reading comprehension dataset containing over 650K question-answer-evidence triples derived from 95K question-answer pairs authored by trivia enthusiasts. It includes evidence documents from Wikipedia and the Web, with an average of six documents per question. The dataset is divided into training (61,888 for Wiki, 76,496 for Web), development (7,993 for Wiki, 9,951 for Web), and test sets (7,701 for Wiki, 9,509 for Web). TriviaQA is designed to test models' ability to handle complex questions with significant lexical and syntactic variability, requiring multi-sentence reasoning.
\end{itemize}
\begin{table}[ht]
\centering
\fontsize{9pt}{9pt}\selectfont
\setlength{\tabcolsep}{6pt}
\renewcommand{\arraystretch}{1.1}
\begin{tabular}{l|c|c|c}
\toprule
\rowcolor{gray!10}
\textbf{Dataset} & \textbf{Train} & \textbf{Dev} & \textbf{Test} \\
\midrule
2Wiki& 154,878 & 12,576 & 12,576 \\
HotpotQA  & 90,564 & 7,405 & 14,810 \\
Bamboogle  & 0 & 0 & 125 \\
FRAMES & 0 & 0 & 824 \\
MuSiQue  & 19,938 & 2,417 & 2,459 \\
NQ  & 307,373 & 7,830 & 7,842 \\
TriviaQA& 138,384 &17,944&17,210 \\
\bottomrule
\end{tabular}
\caption{Overview of datasets used in experiments.}
\label{tab:dataset_overview}
\end{table}

\section{Baseline Implementation Details}
\label{appendix:baselines}
To comprehensively evaluate the effectiveness of the proposed method, we compare it against a diverse set of strong baselines covering different paradigms of reasoning, retrieval, and agentic interaction. The implementation details are described below.

\paragraph{\ding{182} No-RAG.}
This category evaluates the model's capability when no retrieval mechanism is involved, serving to assess how well the LLM performs using only its internal parametric knowledge.
\begin{itemize}
    \item \textbf{No-RAG.}
    A non-retrieval baseline where the LLM directly generates responses based solely on its parametric knowledge, without access to any external information.
    \item \textbf{CoT (Chain-of-Thought)~\cite{wei2022chain}.}
    Following prior work, this baseline prompts the model to generate explicit intermediate reasoning steps before producing the final answer.
\end{itemize}
\paragraph{\ding{183} Naive RAG.}
Naive RAG methods incorporate retrieval in a heuristic or fixed manner without explicit reasoning–retrieval coordination, aiming to provide external knowledge through simple retrieval strategies.
\begin{itemize}
    \item \textbf{FS-RAG (Fix Sentence RAG)~\cite{trivedi2023interleaving}.}
    FS-RAG performs retrieval at the sentence level, where each sentence in the input context is independently used as a retrieval query to obtain relevant external knowledge.
    \item \textbf{FL-RAG (Fix Length RAG)~\cite{khandelwal2020generalizationmemorizationnearestneighbor}.}
    FL-RAG invokes retrieval at fixed token intervals, triggering the retrieval module every (n) generated tokens using the preceding token window as the query.
\end{itemize}
\paragraph{\ding{184} Agentic RAG.}
Agentic RAG methods rely on manually designed prompting strategies to interleave reasoning and retrieval, enabling the model to decide when and how to invoke external tools during inference.
knowledge through simple retrieval strategies.
\begin{itemize}
    \item \textbf{ReAct~\cite{react}.}
    ReAct interleaves reasoning and action steps, enabling the model to interact with external tools or environments during inference.
    \item \textbf{IRCOT~\cite{trivedi2023interleaving}.}
   IRCOT alternates between retrieval and Chain-of-Thought generation, where intermediate reasoning guides retrieval and retrieved evidence is used to refine reasoning.
    \item  \textbf{TC-RAG~\cite{jiang2024tcrag}.}
    TC-RAG introduces a Turing-complete RAG framework 
    with a memory stack, allowing the model to dynamically control retrieval through push and pop. 
\end{itemize}
\paragraph{\ding{185} RL-based Agentic RAG.}
RL-based Agentic RAG approaches apply RL to retrieval-augmented settings, training models to learn when and how to invoke external retrieval tools during multi-step reasoning.
\begin{itemize}
    \item \textbf{ReSearch~\cite{chen2025learning}.}
    ReSearch is trained using RL to jointly perform reasoning and search for multi-hop question answering, without requiring supervision over intermediate reasoning steps. We implement this baseline using Proximal Policy Optimization (PPO)~\cite{schulman2017proximal}.
    \item \textbf{Search-R1~\cite{jin2025search}.}
    Search-R1 trains the LLM via RL to autonomously interact with a search engine during reasoning. It generates multi-round search queries and integrates retrieved information for problem solving. We train this baseline using Group Relative Policy Optimization (GRPO)\cite{shao2024deepseekmath}.
    \item  \textbf{AEPO~\cite{dong2025agentic}.}
    AEPO adopts an entropy-balanced policy optimization framework to stabilize training in high-entropy agentic settings by balancing entropy during rollout generation and policy updates.
    \item  \textbf{ARPO~\cite{dong2025agentic2}.}
   ARPO introduces an entropy-aware adaptive rollout strategy that dynamically adjusts sampling at high-entropy decision points, encouraging diverse tool-use behaviors.
    \item  \textbf{Mem1~\cite{zhou2025mem1}.}
    Mem1 is an end-to-end RL framework for long-horizon multi-turn tasks under constant memory constraints, maintaining a compact internal state for reasoning and memory consolidation.
\end{itemize}



\section{Implementation Details}
\label{appendix:Implementation}
We implement our method using a GRPO-based reinforcement learning framework. Training is performed with a batch size of 16 and a learning rate of $1\times10^{-6}$, without quantization or KV caching. We train for a single GRPO epoch consisting of 300 optimization steps. For each prompt, four stochastic generations are sampled with a temperature of 0.7, diverse sampling enabled, and a diversity penalty of 1.0. Each generation is limited to 2048 new tokens, with a maximum aggregated context length of 4096 tokens and up to ten iterative generation steps per episode. Rewards are computed at the token level. The GRPO optimizer follows a clipped objective with $\beta=0.04$, $\mu=2$, and $\epsilon=0.1$. All hyperparameters are fixed across experiments, as training is conducted on a single curated dataset (Table~\ref{tab:impl_hyperparams}).

\begin{table}[ht]
\centering
\fontsize{9pt}{9pt}\selectfont
\renewcommand{\arraystretch}{1.2}
\setlength{\tabcolsep}{6pt}
\begin{tabular}{l|c}
\toprule
\rowcolor{gray!10}
\textbf{Hyperparameter} & \textbf{Value} \\
\midrule
Batch Size & 16 \\
Learning Rate & $1\times10^{-6}$ \\
GRPO Epochs & 1 \\
Steps per Epoch & 300 \\
Number of Generations & 10 \\
Max New Tokens & 2{,}048 \\
Max Context Length & 4{,}096 \\
Max Generate Iterations & 10 \\
Temperature & 0.7 \\
$\beta$ (GRPO) & 0.04 \\
$\mu$ (GRPO) & 2 \\
$\epsilon$ (GRPO) & 0.1 \\
\bottomrule
\end{tabular}
\caption{Implementation details and hyperparameter settings used in all experiments.}
\label{tab:impl_hyperparams}
\end{table}

\section{Prompt}
\label{sec: prompts} 
In this section, we provide the prompts used in our framework:

\begin{tcolorbox}[
  colback=lightgray!20,
  colframe=darkgray!80,
  title=\textsc{\M~Prompt Template}
]
\small
\textbf{Answer the following question to the best of your ability.  
You have access to the tools listed below.}

\medskip
\textit{[Insert \texttt{tool\_description}s here]}

\medskip
The reasoning process and final answer should be wrapped in the <ACTION> and </ACTION> tags, respectively:  \\
1. That is, "<Think> Write the reasoning process here </Think>"; \\
2. "<Conclusion> Write the final answer here </Conclusion>"; \\
3. Use the <Backtrack> tag to wrap your reflection results;  \\
4. Use the <Summary> tag to wrap your summary results;  \\
5. During the thinking process, **if necessary, the assistant can perform a search** to find uncertain knowledge, formatted as "<Search> search query. The search system will then provide the assistant with the retrieved information in the format "<Observation> ...search results... </Observation>".

\medskip
Please think strictly according to the described process and use the following
\textbf{Action Format}:

\begin{flushleft}
\texttt{User\_Query:} … \\
\texttt{<Think>} … \texttt{</Think>} \\
\texttt{<Search>} … \texttt{</Search>} \\
\texttt{<Summary>} … \texttt{</Summary>} \\
\texttt{<Backtrack>} … \texttt{</Backtrack>} \\
…(repeat \texttt{<Think>} / \texttt{<Search>} / \texttt{<Summary>} / \texttt{<Backtrack>} as needed) \\
\texttt{<Think>} I now know the final answer. \texttt{</Think>} \\
\texttt{<Conclusion>} … \texttt{</Conclusion>} \\
\end{flushleft}

\textbf{Begin!}
\end{tcolorbox}

\begin{tcolorbox}
[colback=lightgray!20,colframe=darkgray!80,title=GRM Relevance Evaluation Prompt]
\small
\label{tab:grm_relevance_prompt}

You are a relevance evaluator. Your task is to assess the semantic relevance of the retrieved documents with respect to the given question and task context, and output a relevance score between \textbf{0.0} and \textbf{1.0}, reported to \textbf{one decimal place} (e.g., 0.7).

\medskip
\textbf{Question:}  
\texttt{\{Question\}}

\textbf{Search Arguments:}  
\texttt{\{Search Args\}}

\textbf{Retrieved Documents:}  
\texttt{\{Search Docs\}}

\medskip
\textbf{\textit{Scoring Guidelines:}}
\begin{itemize}[leftmargin=*,noitemsep]
    \item \textbf{1.0}: Highly relevant; directly answers the question and strongly supports reasoning.
    \item \textbf{0.7--0.9}: Mostly relevant and useful, with minor noise or redundancy.
    \item \textbf{0.4--0.6}: Partially relevant; indirect, incomplete, or weakly aligned.
    \item \textbf{0.1--0.3}: Largely irrelevant or misleading, with only marginal relevance.
    \item \textbf{0.0}: Completely irrelevant to the question.
\end{itemize}

\medskip
\textbf{Response Format:}

\texttt{Relevance Score:}

\end{tcolorbox}

\begin{tcolorbox}[
  colback=lightgray!20,
  colframe=darkgray!80,
  title=GRM Memory Evaluation Prompt
]
\small

You are a memory-operation evaluator. Your task is to assess whether
the agent's memory operation is reasonable and beneficial under the
current question and stack context, and output a quality score between
\textbf{0.0} and \textbf{1.0}, reported to \textbf{one decimal place}
(e.g., 0.7).

\medskip
\textbf{Question:}  
\texttt{\{Question\}}

\textbf{Current Reasoning Context:}  
\texttt{\{Reasoning Context\}}

\textbf{Memory Stack Before Operation:}  
\texttt{\{Stack Before\}}

\textbf{Executed Memory Action:}  
\texttt{\{Memory Action\}}

\textbf{Memory Stack After Operation:}  
\texttt{\{Stack After\}}

\medskip
\textbf{\textit{Scoring Guidelines:}}
\begin{itemize}[leftmargin=*,noitemsep]
    \item \textbf{1.0}: Highly reasonable and highly beneficial; the
    operation clearly improves the current reasoning state. For
    \texttt{<Backtrack>}, it removes misleading, redundant, or
    no-longer-useful memory while preserving useful information. For
    \texttt{<Summary>}, it effectively compresses prior context while
    preserving key facts needed for future reasoning.
    \item \textbf{0.7--0.9}: Mostly reasonable and useful, with only minor
    information loss, redundancy, or imperfect compression/removal.
    \item \textbf{0.4--0.6}: Partially reasonable; the intention is
    acceptable, but the operation is only weakly helpful, incomplete, or
    somewhat misaligned with the current reasoning needs.
    \item \textbf{0.1--0.3}: Largely unreasonable or weakly beneficial;
    the operation removes useful information, keeps misleading content,
    or produces an unhelpful summary.
    \item \textbf{0.0}: Completely unreasonable or harmful; the operation
    clearly damages the current reasoning state or is irrelevant to the
    task.
\end{itemize}

\medskip
\textbf{\textit{Evaluation Principles:}}
\begin{itemize}[leftmargin=*,noitemsep]
    \item Judge the operation with respect to the current question and
    reasoning state.
    \item Reward operations that help correct wrong branches, reduce
    distraction, or compress context without losing key evidence.
    \item Penalize operations that discard important facts, preserve
    misleading branches, or produce vague/inaccurate summaries.
    \item Focus on whether the operation is reasonable and beneficial at
    this step, not whether the final answer is ultimately correct.
\end{itemize}

\medskip
\textbf{Response Format:}

\texttt{Memory Operation Score:}

\end{tcolorbox}

\section{Reward Analysis}
\label{sec:reward_analysis}
Our process rewards ($r_{\text{rag}}$ and $r_{\text{mem}}$) are scored by a general-purpose LLM judge (GRM) via a fixed evaluation prompt, requiring no additional reward model training.

\paragraph{Semantic Reward Analysis.}
To assess whether semantic process evaluation is necessary, we replace the LLM judge with heuristic rewards: ROUGE-L for retrieval relevance and a positive score for any valid memory operation. The average F1 score across seven benchmarks decreases from 33.55 to 29.40, and further drops to 19.94 when both $r_{\mathrm{rag}}$ and $r_{\mathrm{mem}}$ are removed, demonstrating the importance of semantic process supervision.

\section{Additional Experiments}
\label{sec:additional experiments}

\begin{table}[ht]
\centering
\fontsize{9pt}{10pt}\selectfont
\renewcommand{\arraystretch}{1}{
\setlength{\tabcolsep}{8pt}

\begin{tabular}{l|cc}
\toprule
\rowcolor{gray!10}
\textbf{Metric} & \textbf{TC-RAG} & \textbf{\M} \\
\midrule
Avg.Time (s)        & 11.86  & 10.15  \\
Avg.ResponseLen & 838.14 & 880.89 \\
Avg.Actions       & 10.14  & 11.38  \\
Tool-Call Ratio (\%)      & 91.70 & 97.50 \\
\bottomrule
\end{tabular}}
\caption{Statistical comparison of inference latency, response complexity, and tool-use behaviors between TC-RAG and \M~on Qwen2.5-3B.}
\label{tab:statistical_analysis}
\end{table}

\paragraph{\ding{182} Inference Behavior and Tool Usage.}
Table~\ref{tab:statistical_analysis} compares \M with TC-RAG using Qwen2.5-3B. \M reduces average inference latency from 11.86 to 10.15 seconds (14.4\%), while increasing the average number of actions by 12.2\% and response length by 5.1\%. Its tool-call rate increases from 91.7\% to 97.5\%, a difference of 5.8 percentage points. Thus, under this evaluation setting, the additional interactions do not result in higher measured average latency, although they produce more actions and slightly longer responses.





\begin{figure*}[t]
\centering
\setlength{\tabcolsep}{1pt}
\renewcommand{\arraystretch}{1.03}

\begin{tabular}{ccccc}
\toprule
 & \textbf{Reward} & \textbf{Entropy} & \textbf{Resp. Len.} & \textbf{\#Turns} \\
\midrule

{\tiny\textbf{Rej.=0.2}} &
\includegraphics[width=0.22\linewidth]{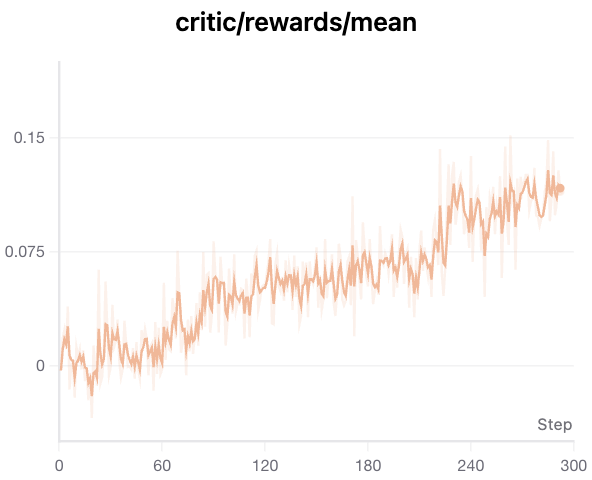} &
\includegraphics[width=0.22\linewidth]{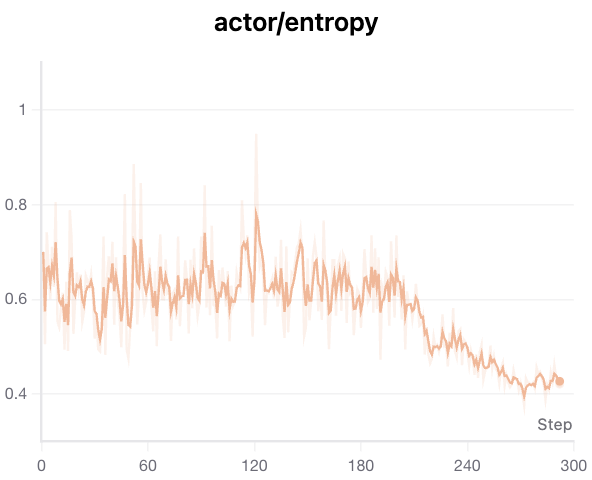} &
\includegraphics[width=0.22\linewidth]{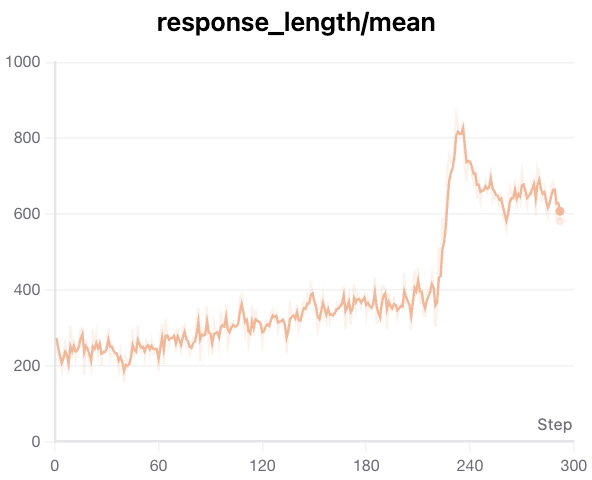} &
\includegraphics[width=0.22\linewidth]{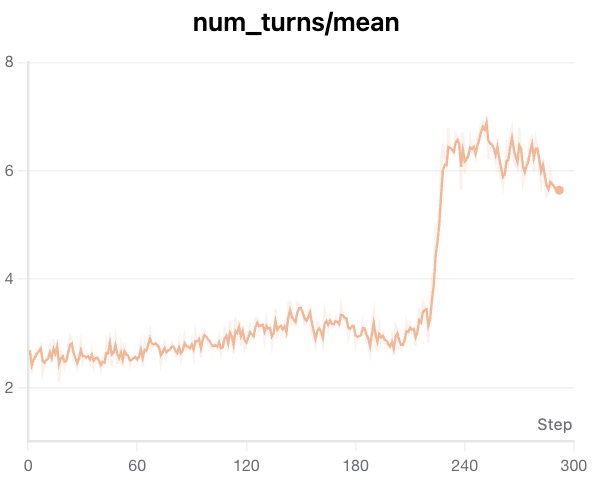} \\

{\tiny\textbf{Rej.=0.4}} &
\includegraphics[width=0.22\linewidth]{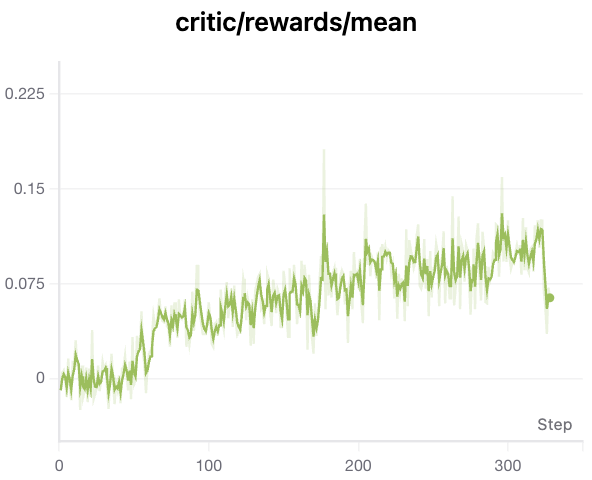} &
\includegraphics[width=0.22\linewidth]{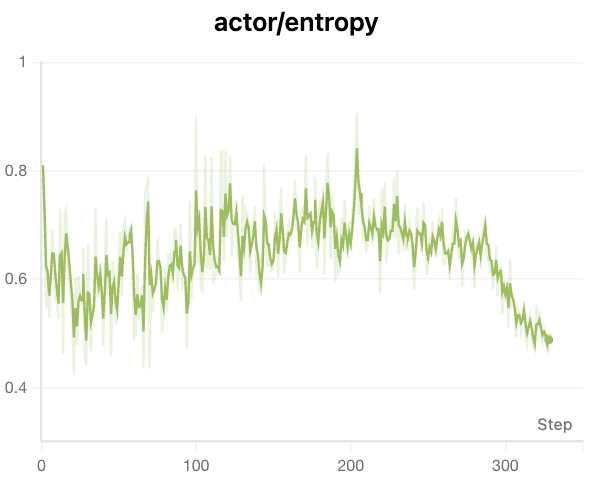} &
\includegraphics[width=0.22\linewidth]{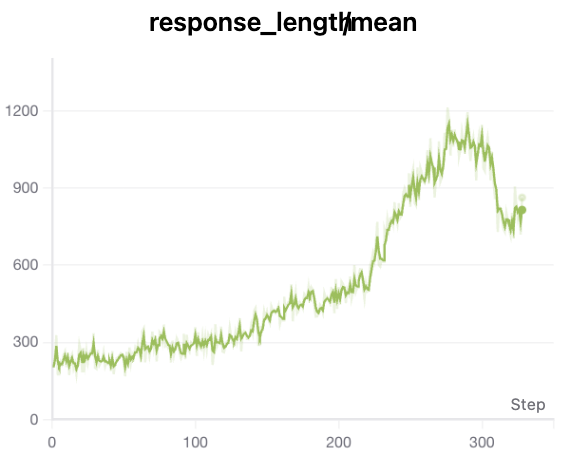} &
\includegraphics[width=0.22\linewidth]{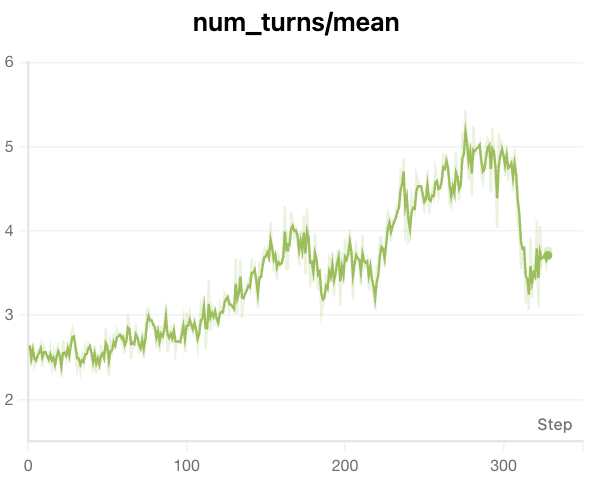} \\

{\tiny\textbf{Rej.=0.6}} &
\includegraphics[width=0.22\linewidth]{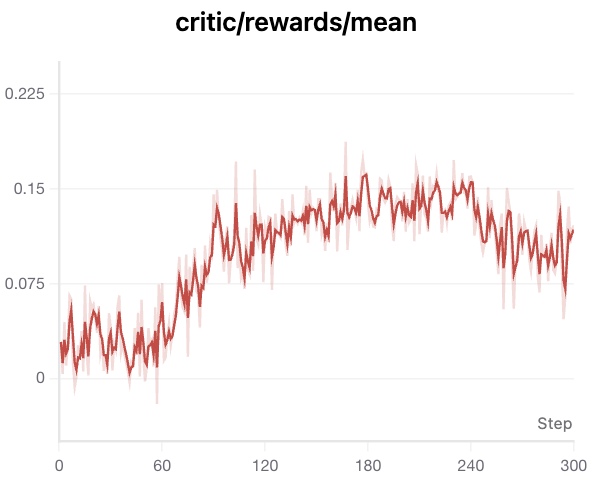} &
\includegraphics[width=0.22\linewidth]{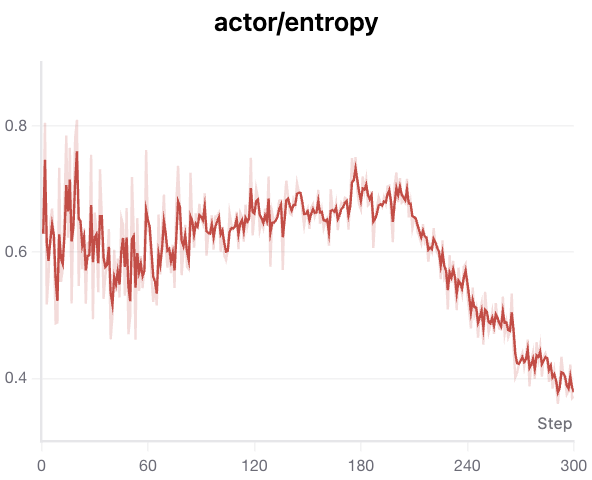} &
\includegraphics[width=0.22\linewidth]{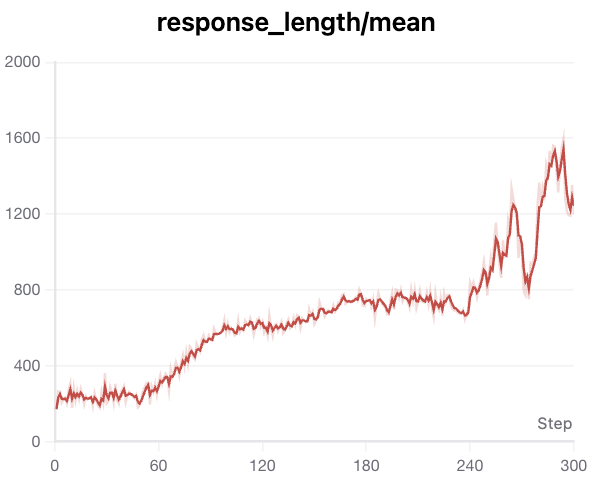} &
\includegraphics[width=0.22\linewidth]{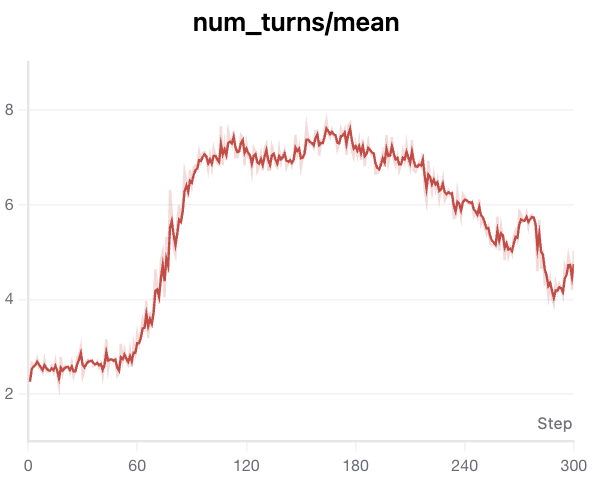} \\

{\tiny\textbf{Rej.=0.8}} &
\includegraphics[width=0.22\linewidth]{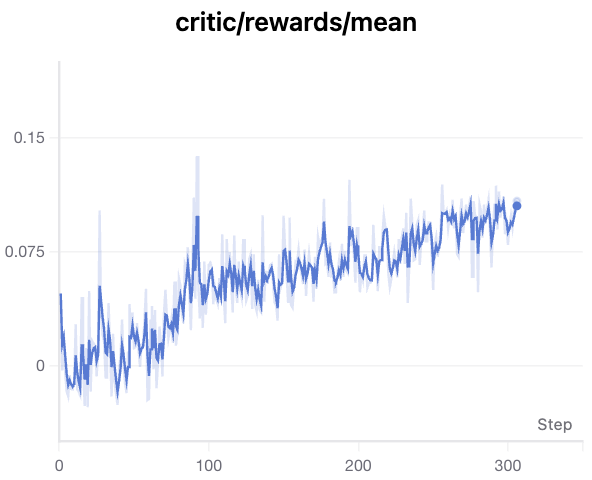} &
\includegraphics[width=0.22\linewidth]{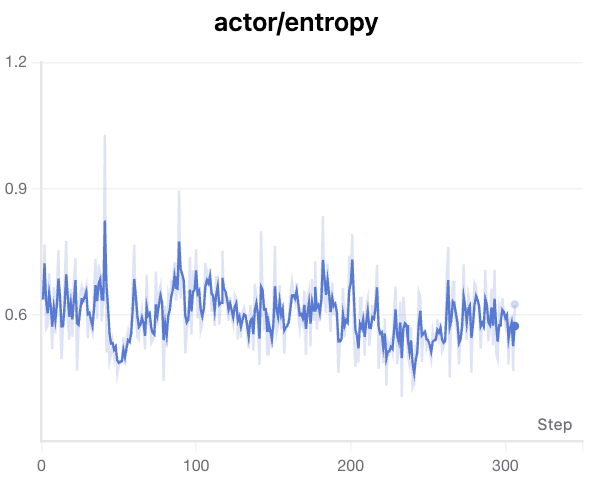} &
\includegraphics[width=0.22\linewidth]{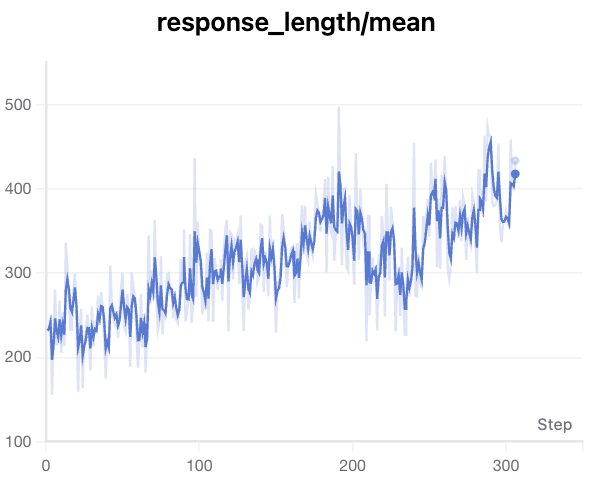} &
\includegraphics[width=0.22\linewidth]{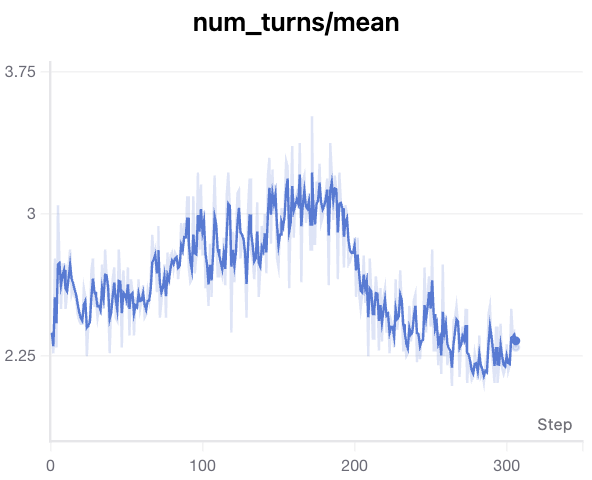} \\

{\tiny\textbf{No Rej.}} &
\includegraphics[width=0.22\linewidth]{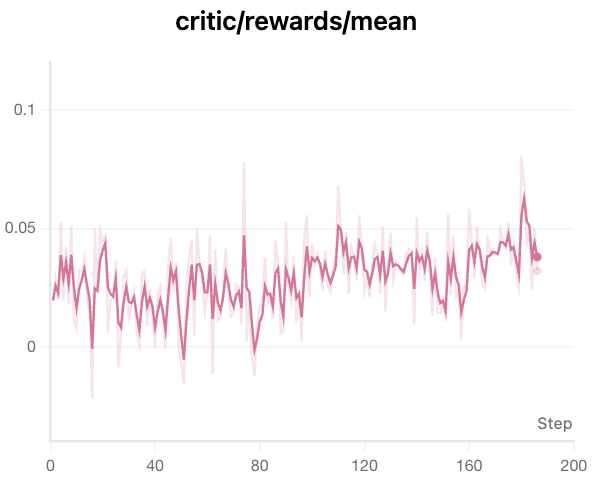} &
\includegraphics[width=0.22\linewidth]{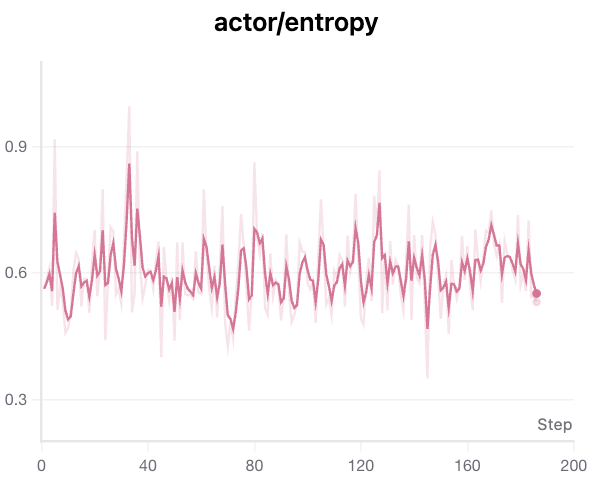} &
\includegraphics[width=0.22\linewidth]{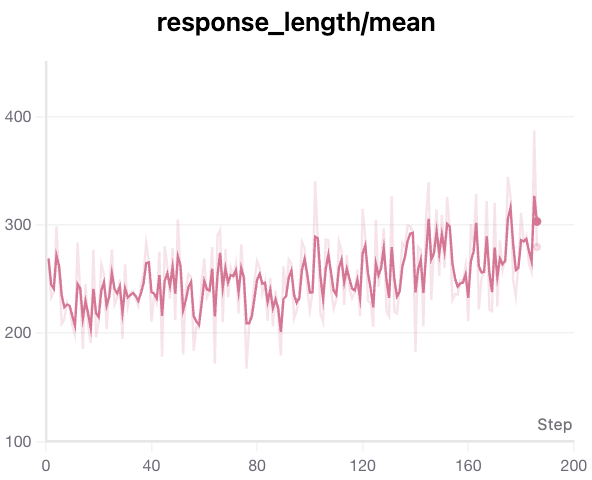} &
\includegraphics[width=0.22\linewidth]{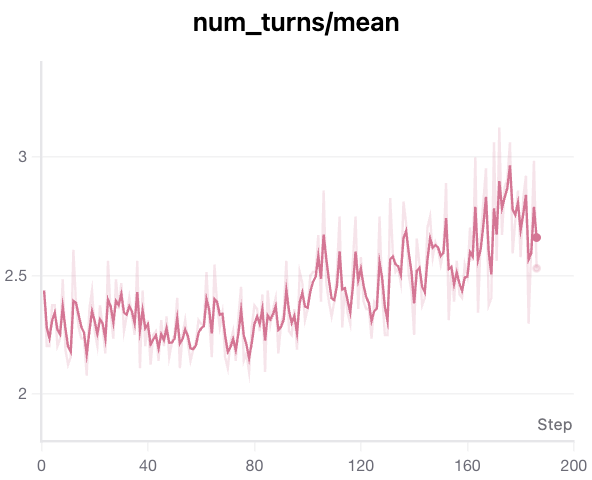} \\

\bottomrule
\end{tabular}

\caption{
Training dynamics of \M~during GRPO optimization on Qwen2.5-3B under different rejection ratios.
}
\label{fig:wandb_all}
\end{figure*}

\paragraph{\ding{183} Action-Specific Importance Analysis}
\begin{figure*}[t]
  \centering
  \includegraphics[width=0.32\linewidth]{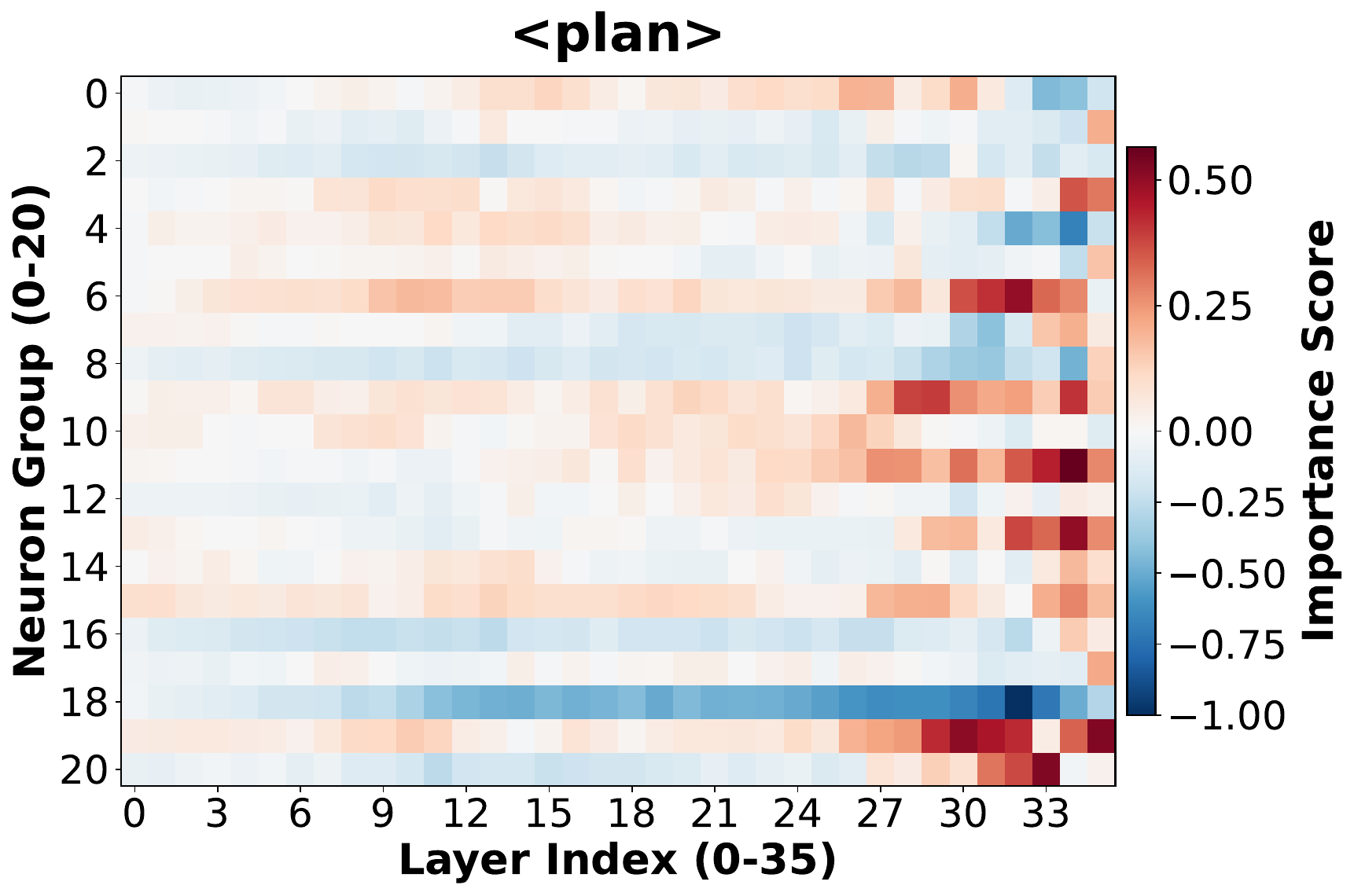}
  \includegraphics[width=0.32\linewidth]{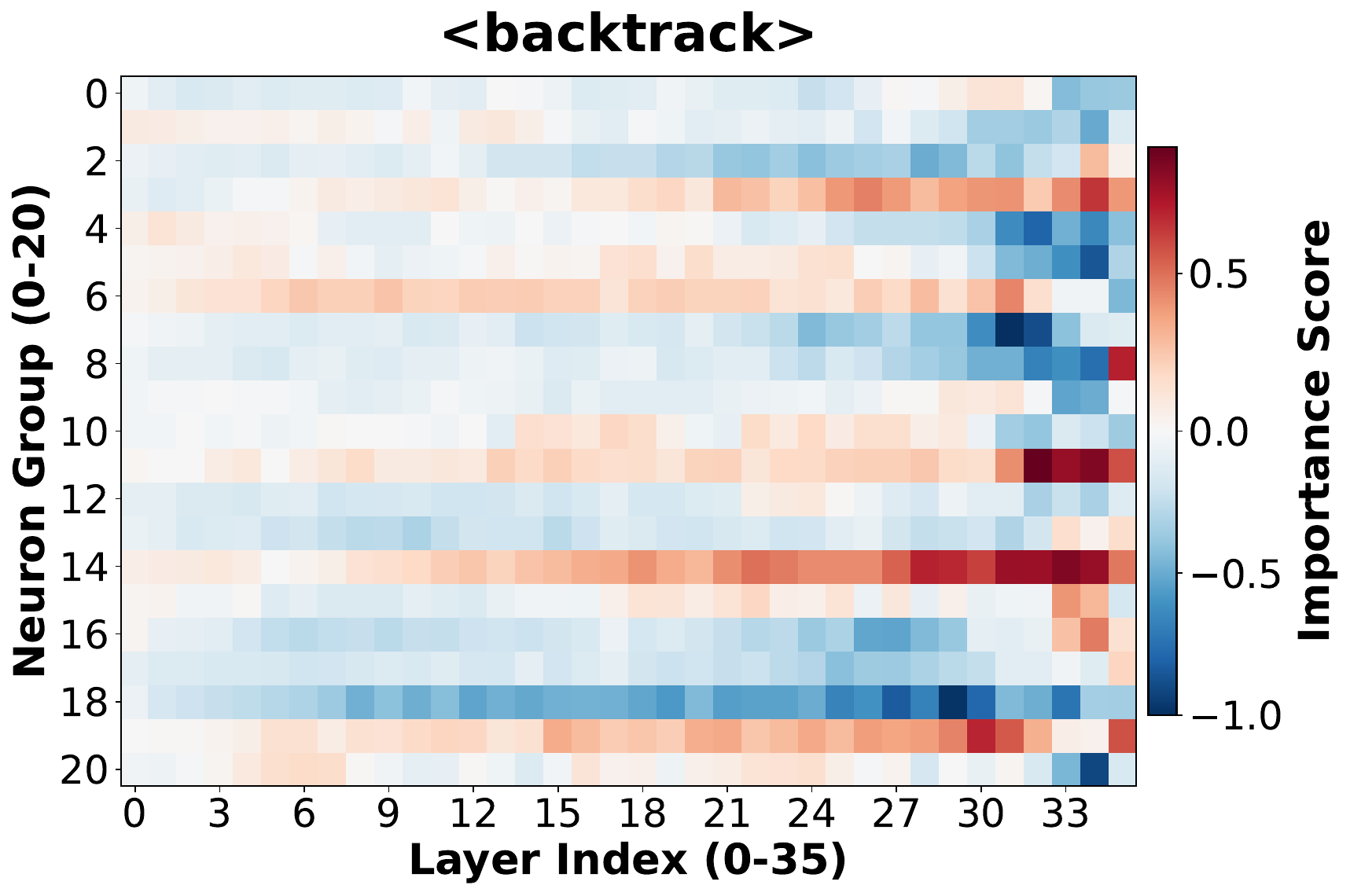}
  \includegraphics[width=0.32\linewidth]{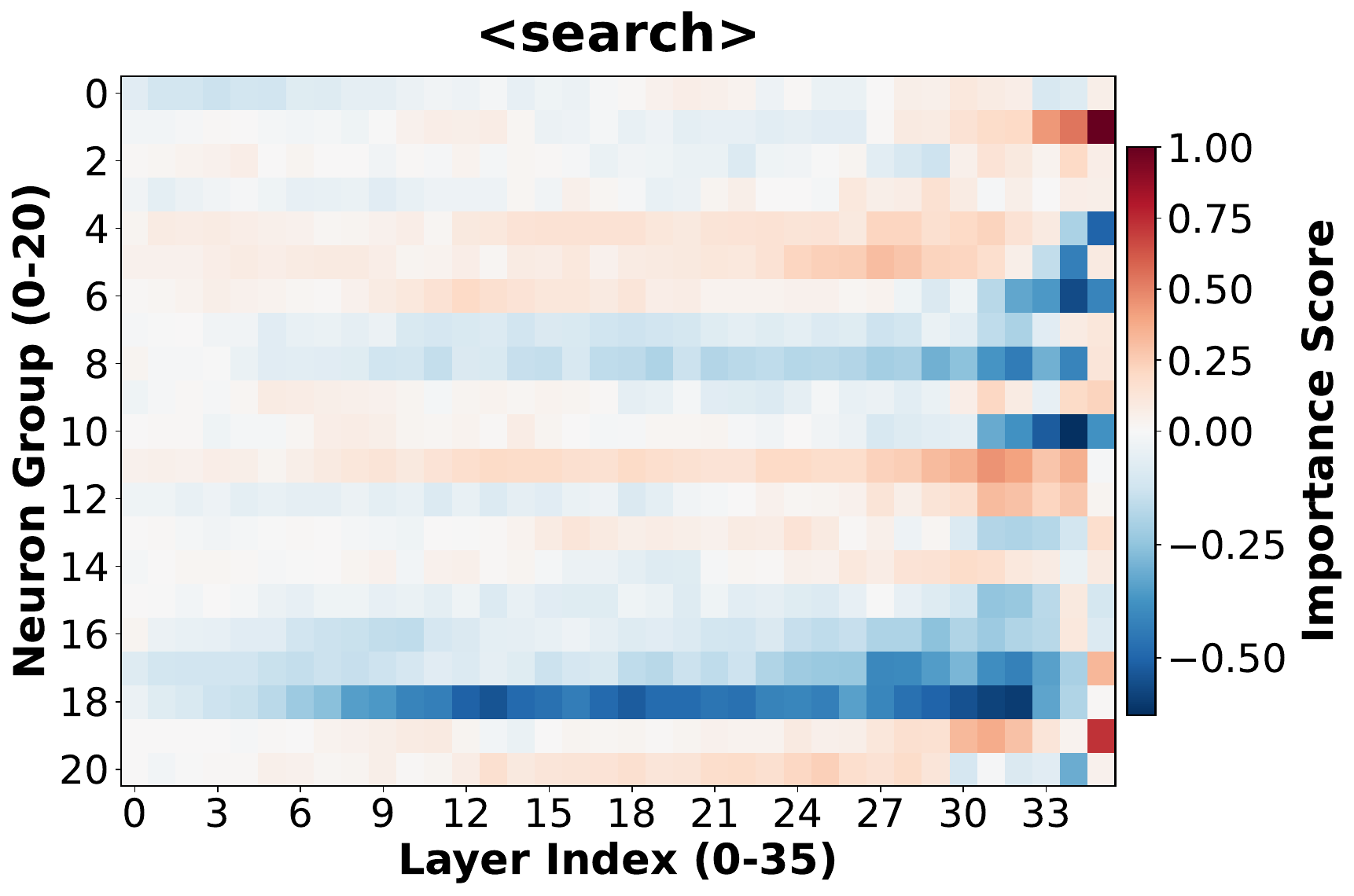}


  \includegraphics[width=0.32\linewidth]{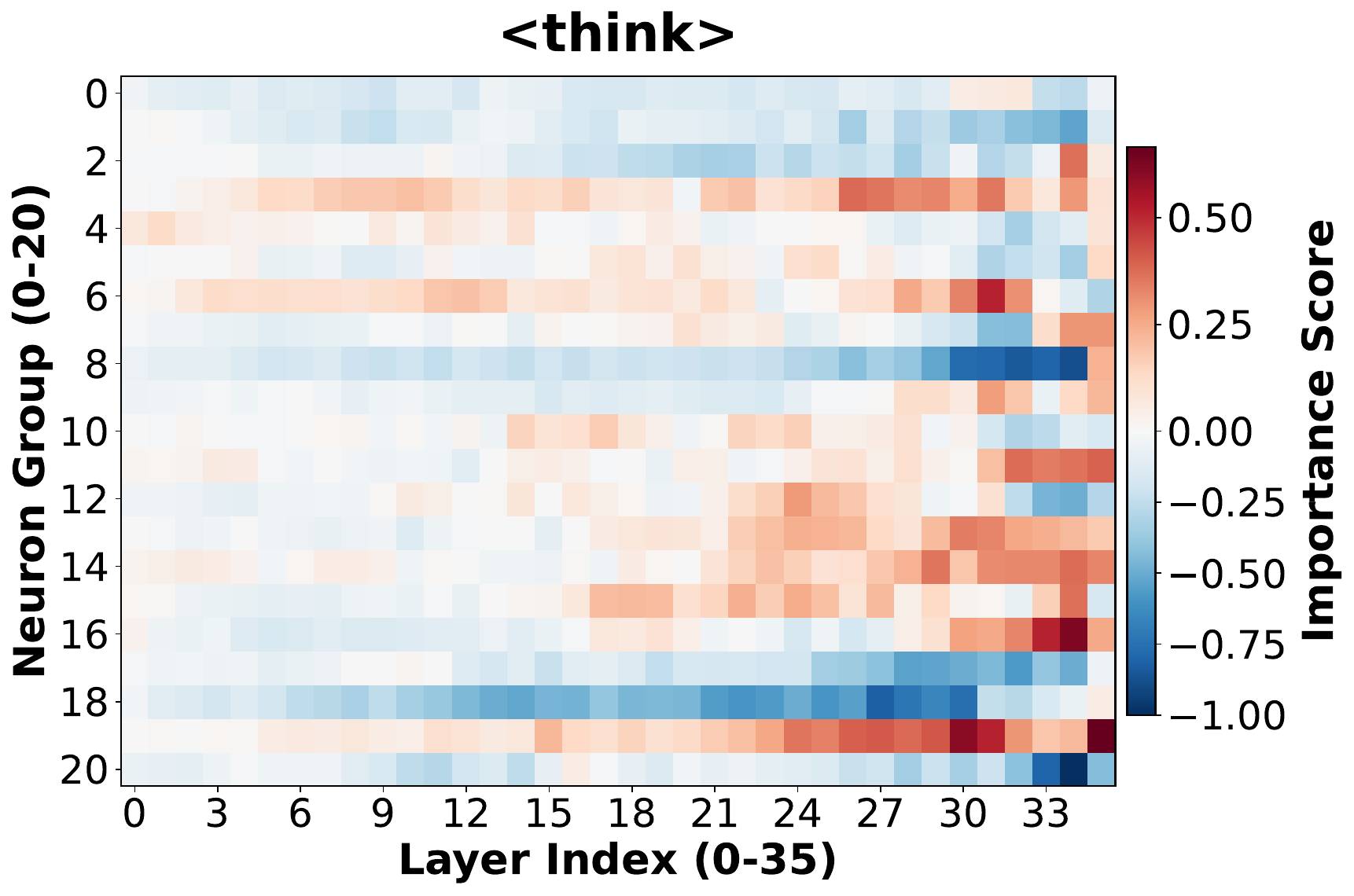}
  \includegraphics[width=0.32\linewidth]{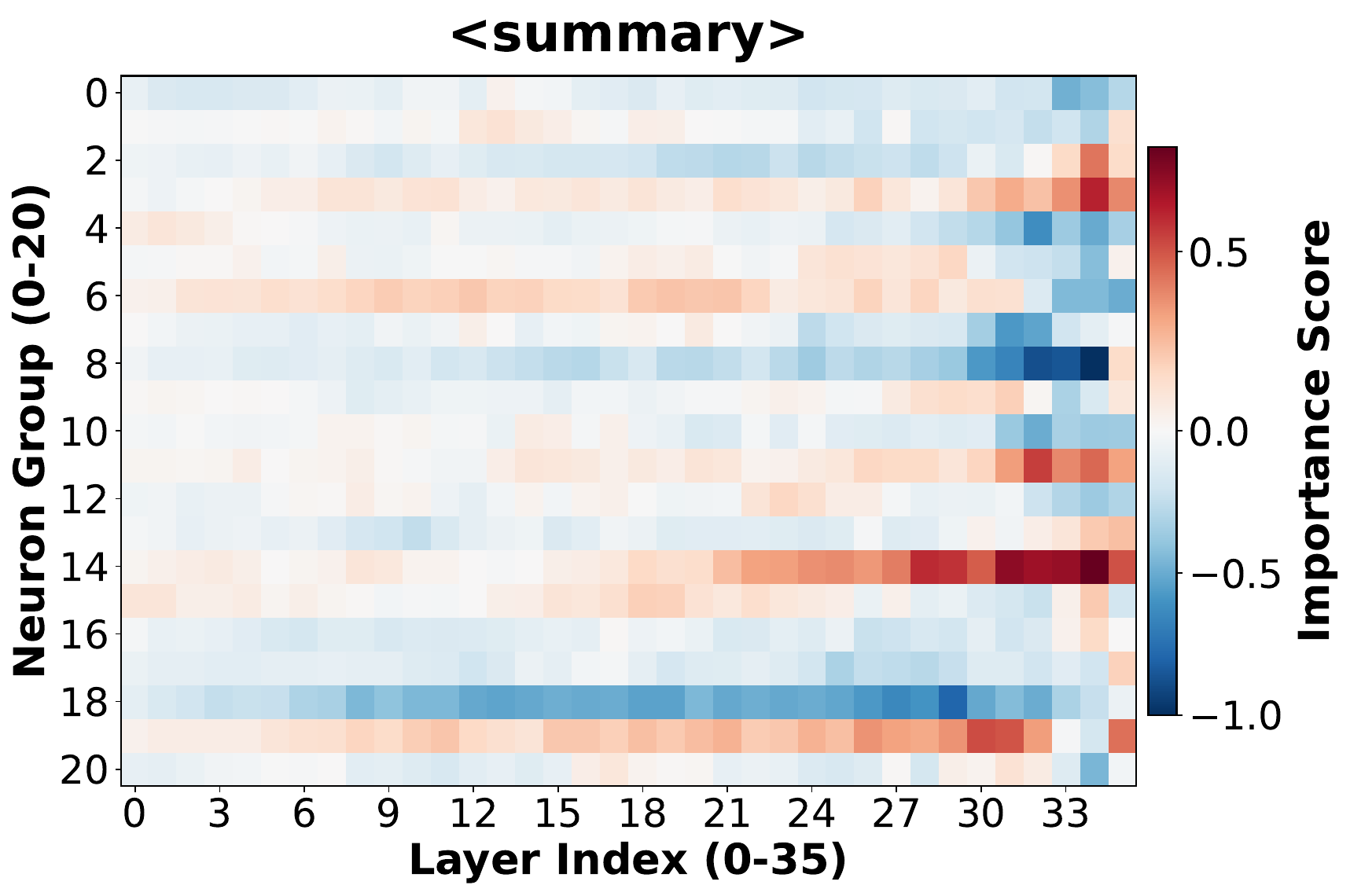}
  \includegraphics[width=0.32\linewidth]{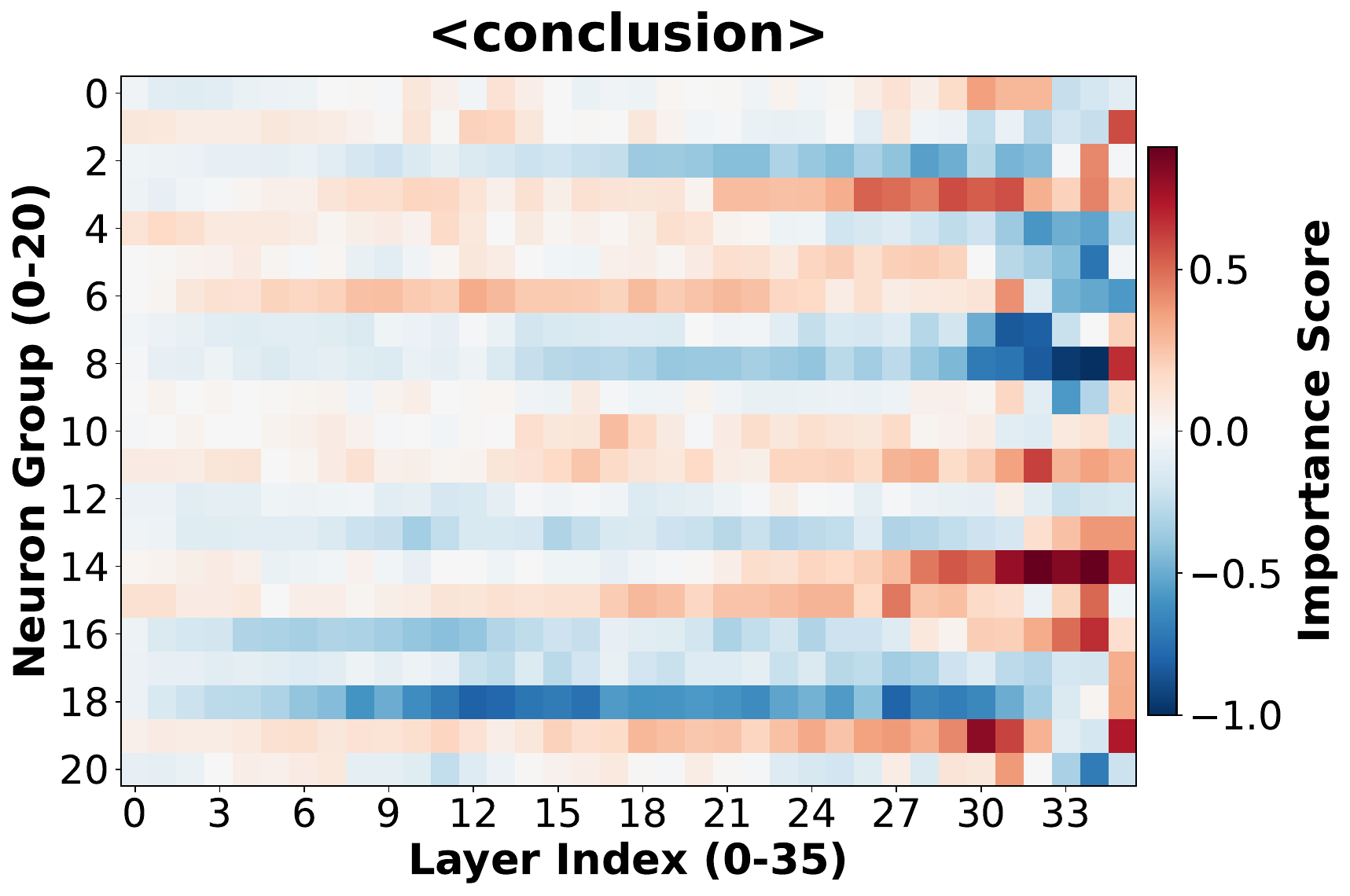}

  \includegraphics[width=0.32\linewidth]{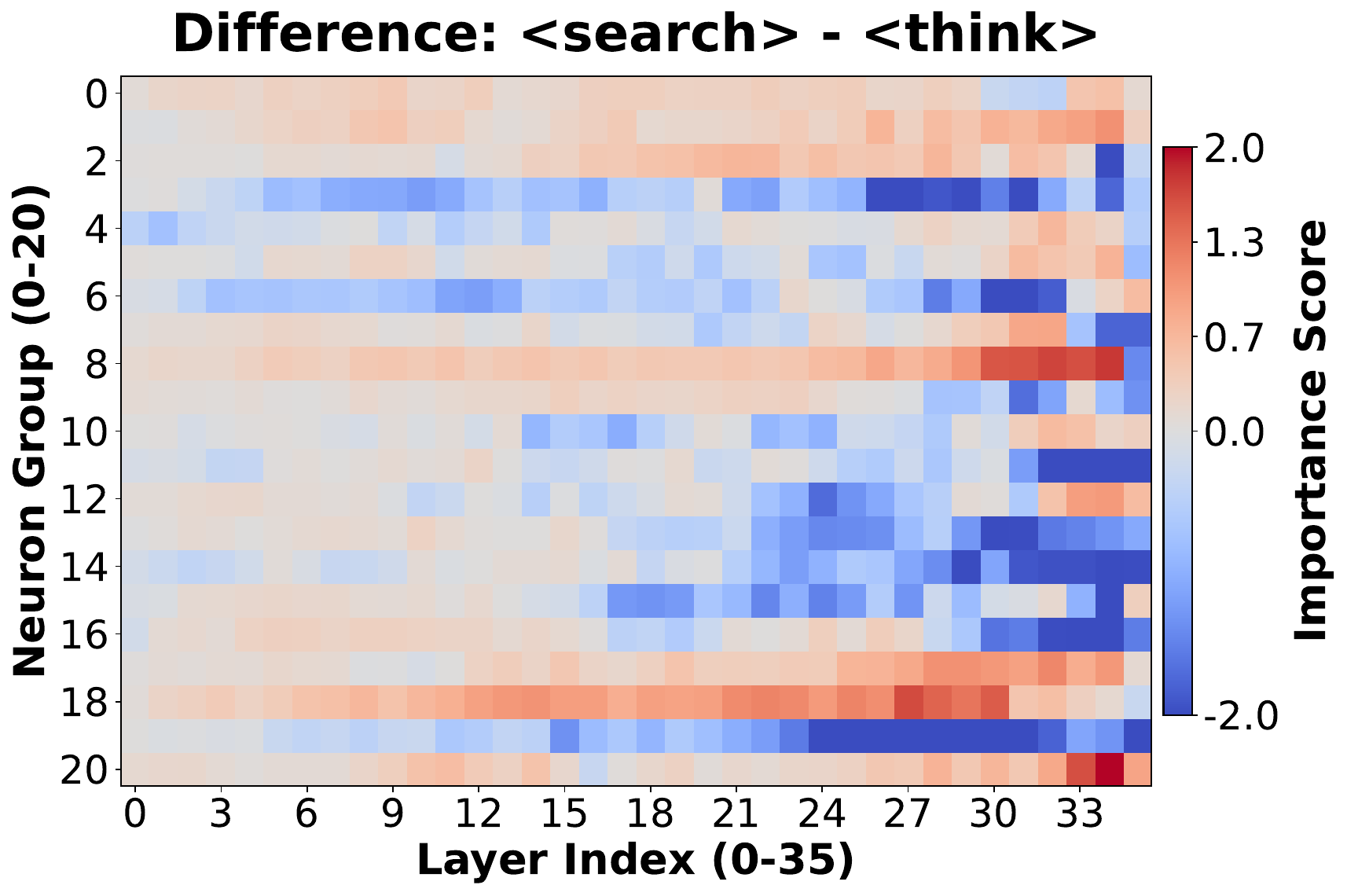}
  \includegraphics[width=0.32\linewidth]{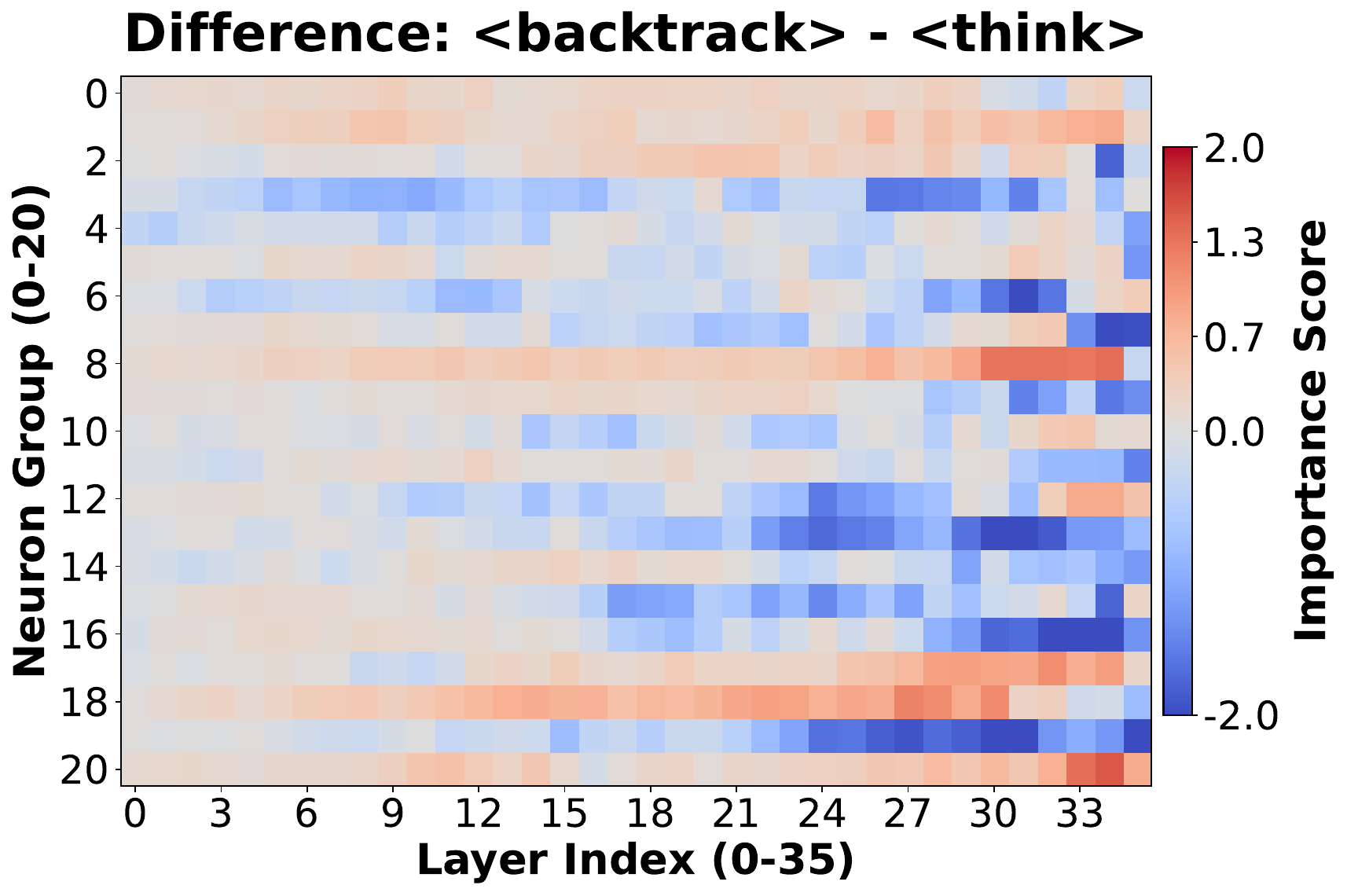}
  \includegraphics[width=0.32\linewidth]{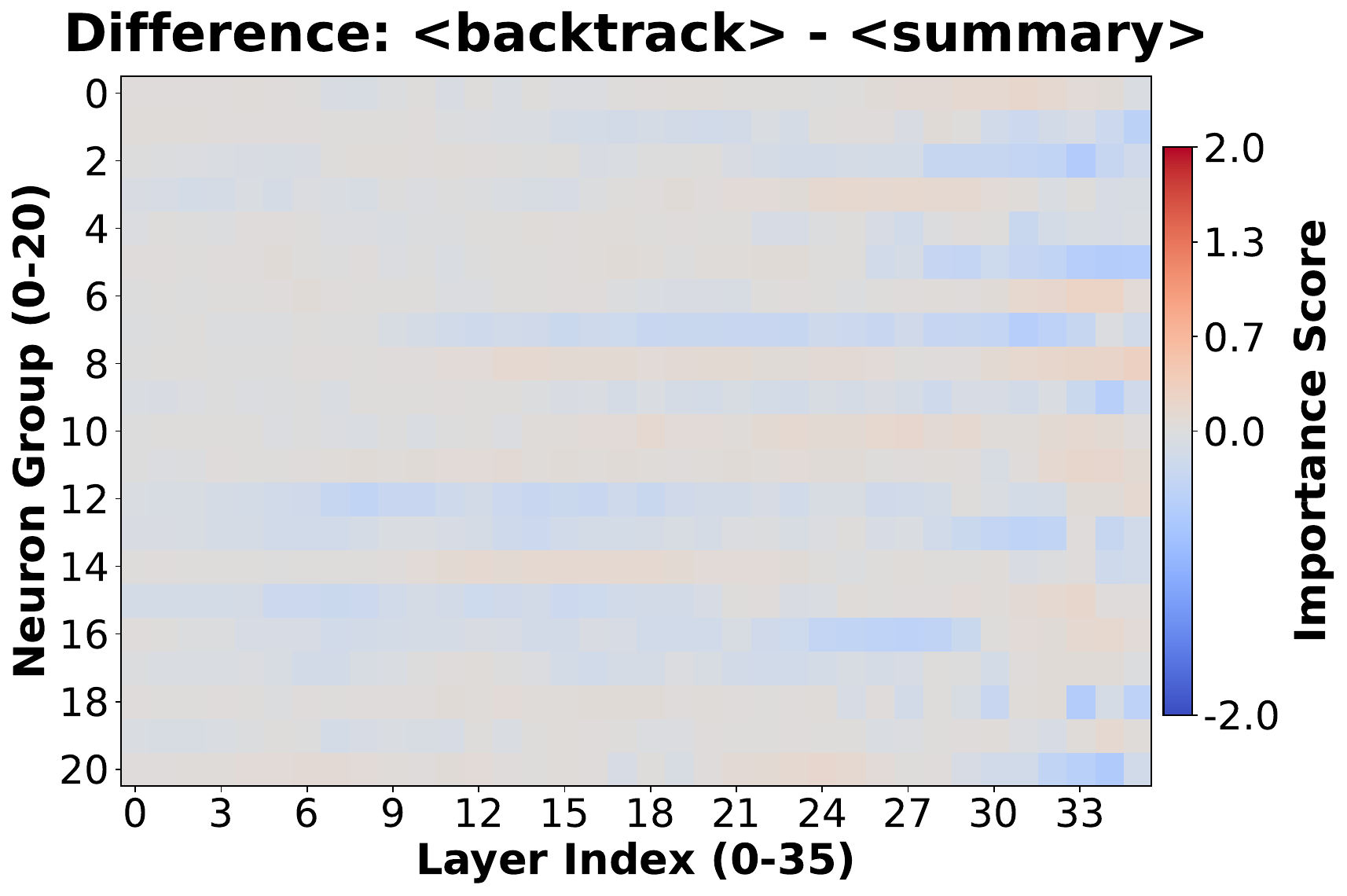}

  \caption{
Neuron-wise importance analysis of different actions. Top two rows: normalized importance scores for individual actions. Bottom row: pairwise differences of unnormalized importance scores between actions.
  }
\label{fig:Importance_Score_heatmap}
\end{figure*}

To analyze how different behaviors engage model parameters~\cite{xu2025parenting,zhang2025knowpo,yang2026dfams}, we define a layer-wise importance score based on neuron activations. For each action token (including both the action token and its argument tokens), we extract its hidden representation at every transformer layer and count the number of neurons with positive activations after the ReLU nonlinearity:
\begin{equation}
    \text{Importance}(l)=\sum \mathbf{ReLU}(\text{embedding}_l>0),
\end{equation}
which measures intensity of neuron participation at layer \(l\). Each embedding is 2048-dimensional and is visualized by averaging every 100 dimensions.

Importance scores are normalized within each action, as shown in the upper panels of Fig.~\ref{fig:Importance_Score_heatmap}. All statistics are computed on a challenging reasoning dataset collected via rejection sampling, retaining only correct-answer trajectories to avoid spurious activations from erroneous reasoning. The top six heatmaps present the \emph{normalized} importance profiles of six actions. At a coarse level, these actions exhibit broadly similar activation patterns, suggesting shared use of the model’s representation space. However, their importance distributions are not identical and encode action-specific differences in both magnitude and layer-wise allocation that are difficult to discern without direct comparison.

To make these distinctions explicit, the bottom row shows pairwise \textbf{differences computed from unnormalized importance scores}. Clear contrasts emerge between \texttt{<Search>} and \texttt{<Think>}, as well as between \texttt{<Backtrack>} and \texttt{<Think>}, indicating substantial divergence in their underlying neuron activations. In contrast, the difference between \texttt{<Summary>} and \texttt{<Backtrack>} is relatively small, suggesting that these memory-related actions rely on highly overlapping neural resources. Overall, the results indicate that tool invocation (\texttt{<Search>}), internal reasoning (\texttt{<Think>}), and memory manipulation (\texttt{<Backtrack>}, \texttt{<Summary>}) correspond to distinct and largely decoupled activation patterns. \textbf{This finding supports explicit action-level decomposition and fine-grained optimization, rather than treating all behaviors as a single homogeneous reasoning process.}

Building on the action-level importance analysis above, we further examine activation patterns at a finer action–argument granularity. As shown in Fig.~\ref{fig:Importance_Score_tag_args}, we separately compute layer-wise importance scores for action tokens and their associated argument tokens using the same activation-sum metric as before, but aggregating activations across layers instead of tokens. While the overall trends remain consistent with the action-level results, clear divergences emerge across several contiguous layer ranges (highlighted in red), where different actions exhibit distinct layer-wise distributions for both their action tags and argument tokens. Importantly, although both components are consistently involved across actions, their relative contributions vary across layers in an action-dependent manner. This suggests that different actions are characterized by distinct layer allocation patterns between action tags and arguments, rather than differences in overall activation magnitude alone.

\begin{figure}[t]
  \centering
  \includegraphics[width=0.5\textwidth]{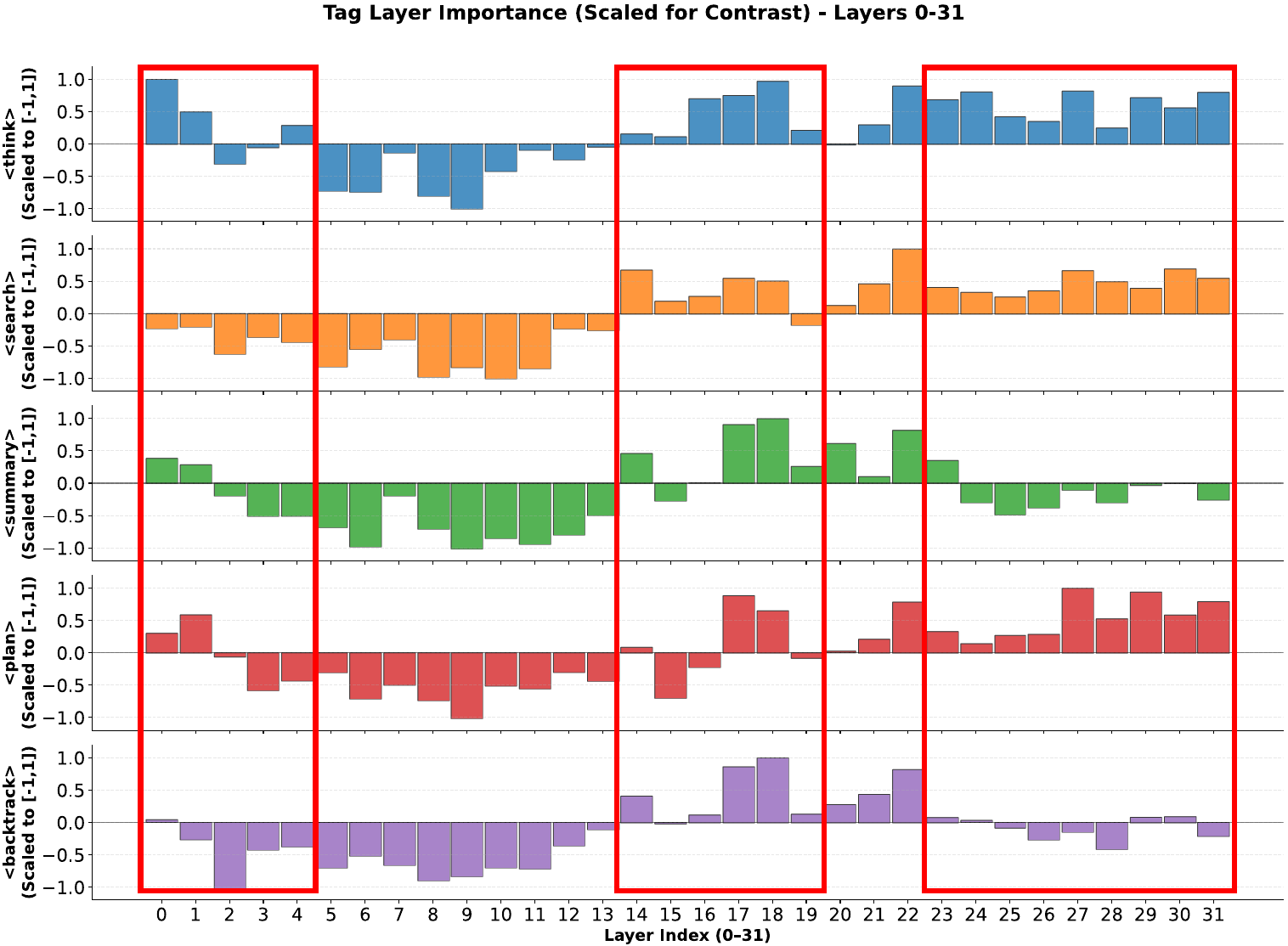}
  \includegraphics[width=0.5\textwidth]{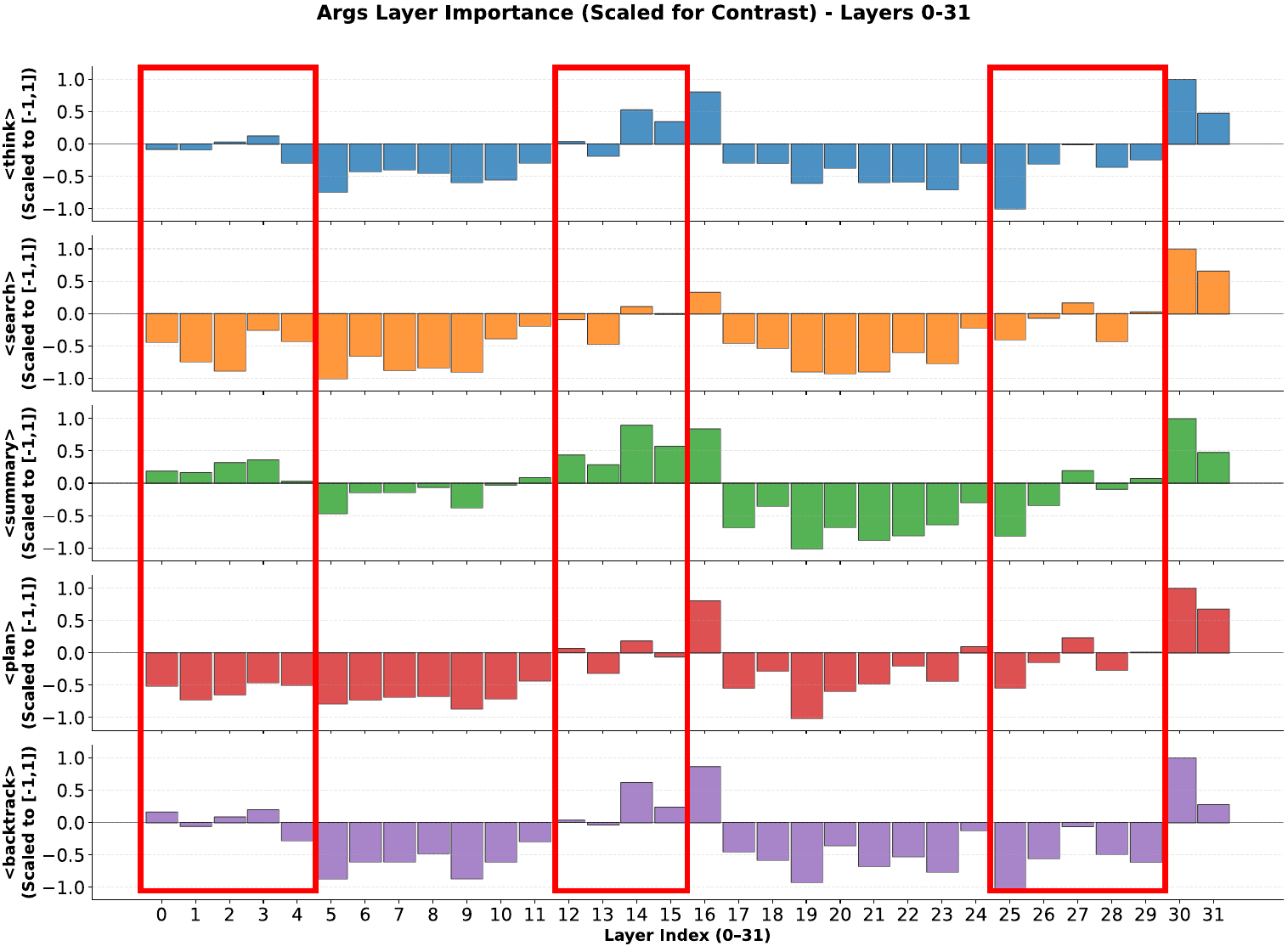}

\caption{
Layer-wise importance profiles of action tokens and their argument tokens.
}

\label{fig:Importance_Score_tag_args}
\end{figure}


\paragraph{\ding{184} Action-Specific Optimization with Self-Search and Decoupling.} 
Motivated by the observation that different actions rely on distinct and only partially overlapping neural resources, we investigate action-specific parameter optimization guided by model Shapley analysis. Specifically, we identify action-relevant parameters using SP-R1 and decouple the optimization by assigning different subsets of parameters to the args layer and action layer, allowing each action to be updated independently. We evaluate this approach on MedQA under two search paradigms: External Search (ES), where both training and testing rely on external Wiki retrieval, and Self-Search (SS)~\cite{fan2025ssrl}, which generates candidate retrievals during training without external knowledge and uses external Wiki search only at test time. This self-supervised search encourages better generalization to unseen questions. Table~\ref{tab:action_decoupling} reports the results: incorporating SS alone improves accuracy over ES (from 65.5\% to 67.5\%), demonstrating its strong effect on generalization, while decoupled, action-specific optimization (Decoup) further boosts performance to 72.2\% under SS and 68.5\% under ES. These results highlight that SS significantly enhances generalization by exposing the model to diverse self-generated retrievals during training, and that Decoup effectively leverages action-specific parameter structure, enabling targeted optimization per action and yielding consistent gains across paradigms.

\begin{table}[t]
\centering
\fontsize{8pt}{10pt}\selectfont
\begin{tabular}{l|lc}
\toprule
\rowcolor{gray!10}
\textbf{Paradigm} & \textbf{Method} & \textbf{Acc (\%)} \\
\midrule
Raw 
& Base & 63.5 \\
& React & 65.5 \\
\midrule
ES
& \M + ES & 66.3 \\
& \M + ES + Decoup & \textbf{68.5} \\
\midrule
SS
& \M + SS & 67.5 \\
& \M + SS+ Decoup & \textbf{72.2} \\
\bottomrule
\end{tabular}
\caption{Effect of action-specific (decoupled) parameter optimization under different search paradigms.}
\label{tab:action_decoupling}
\end{table}

\paragraph{\ding{185} Training-Time Memory-Judge Token Cost.} We quantify the token overhead of memory-operation evaluation during RL training. As shown in Table~\ref{tab:grm_token_cost}, each judge call consumes approximately 3,000--4,000 input tokens and 100--200 output tokens. Across a complete training run, memory-operation evaluation consumes approximately 3.5--4.0 million input tokens and 0.1--0.2 million output tokens. This overhead is incurred only during training and excludes retrieval-relevance judging.

\begin{table}[t]
\centering
\small
\begin{tabular}{lr}
\toprule
\rowcolor{gray!10}
\textbf{Component} & \textbf{Tokens per Call} \\
\midrule
Original context & 2,500--3,000 \\
Generated summary & 300--500 \\
Evaluation prompt and rubric & 200--300 \\
Total input & 3,000--4,000 \\
Output & 100--200 \\
\bottomrule
\end{tabular}
\caption{Estimated token consumption per memory-operation judge call.}
\label{tab:grm_token_cost}
\end{table}

\paragraph{\ding{186} Cross-Domain Rollout Reward Statistics.}
To examine whether the trajectory rejection strategy remains effective across different domains and model scales, we report the mean and variance of rollout rewards. As shown in Table~\ref{tab:rollout_stats}, meaningful rollout variance is preserved across all domains, confirming that the reward distributions do not collapse into indistinguishable ranges.

\begin{table}[h]
\centering
\small
\begin{tabular}{l|cc|cc}
\toprule
\multirow{2}{*}{\textbf{Domain}} & \multicolumn{2}{c|}{\textbf{Qwen2.5-3B}} & \multicolumn{2}{c}{\textbf{Qwen2.5-7B}} \\
& Mean & Var & Mean & Var \\
\midrule
Math & 0.44 & 0.07 & 0.46 & 0.09 \\
Medical & 0.37 & 0.10 & 0.53 & 0.12 \\
General & 0.41 & 0.12 & 0.65 & 0.11 \\
\bottomrule
\end{tabular}
\caption{Rollout reward statistics across domains and model scales.}
\label{tab:rollout_stats}
\end{table}

\paragraph{\ding{187} Scaling Analysis across Model Sizes.}
We compare average performance gains of \M~over TC-RAG across three model scales. Smaller models benefit proportionally more from the framework, while larger models achieve the best absolute performance:
\begin{itemize}[leftmargin=*,noitemsep]
    \item Qwen2.5-1.5B: 17.9 $\rightarrow$ 27.6 (+54\%)
    \item Qwen2.5-3B: 22.2 $\rightarrow$ 33.5 (+50\%)
    \item Qwen2.5-7B: 30.5 $\rightarrow$ 36.6 (+20\%)
\end{itemize}
This pattern suggests that the structured action space and memory mechanisms are particularly effective at compensating for limited parametric capacity in smaller models.

\section{Case Study}

We present a representative case study to illustrate the model’s ability
to perform adaptive retrieval refinement and error correction.
As shown in Table~\ref{tab:case2}, the model initially
issues broad retrieval queries that return irrelevant or noisy information,
including unrelated medical content and off-topic encyclopedic entries.
Instead of committing such noise to its internal state, the model explicitly
invokes backtracking to revise retrieval keys and progressively narrow the
search scope.
Through multiple iterations of retrieval, inspection, and memory revision, the
model correctly determines that \textit{Battleship} is the
2012 military science fiction action film directed by Peter Berg in which Jason
Snell worked.
These cases highlight the effectiveness of fine-grained memory control and
retrieval refinement in mitigating error accumulation during multi-step
reasoning.

\begin{table*}[htbp]
\centering
\begin{tabular}{
  >{\centering\arraybackslash}m{2cm}
  >{\raggedright\arraybackslash}m{11.25cm}
}
\hline
\textbf{Question} & On which 2012 American military science fiction action film, directed by Peter Berg, did Jason Snell work? \\ \hline
\textbf{Ground Truth} & Battleship \\ \hline
\textbf{Output} & 
\begin{tcolorbox}[colback=violet!12!white,colframe=white,boxsep=0pt,left=0pt,right=0pt,top=0pt,bottom=0pt,after=\vspace{-8pt} ]
<Plan>
First, I need to find the 2012 American military science fiction action film directed by Peter Berg. Then, I will search for the actors who worked on the film to identify Jason Snell.
</Plan>
\end{tcolorbox}

\begin{tcolorbox}[colback=green!15!white,colframe=white,boxsep=0pt,left=0pt,right=0pt,top=0pt,bottom=0pt,after=\vspace{-8pt} ]
<Search>
[Wiki\_RAG]:\textbf{ Peter Berg's} films
</Search>
\end{tcolorbox}

\begin{tcolorbox}
[colback=gray!12!white,colframe=white,boxsep=0pt,left=0pt,right=0pt,top=0pt,bottom=0pt,after=\vspace{-8pt} ]
<Observation>
'1. Gerald T. Olson is an American film producer best known for films such as House Party, Bloody Birthday and Dumb $\&$ Dumber. The San Francisco State alum relocated...'
</Observation>
\end{tcolorbox}
\begin{tcolorbox}[colback=violet!12!white,colframe=white,boxsep=0pt,left=0pt,right=0pt,top=0pt,bottom=0pt,after=\vspace{-8pt} ] 
<Backtrack> While the tool did provide some information about films directed by Peter Berg, \textbf{it didn't specify the 2012 military} science fiction action film I was looking for... Let me search for the 2012 military science fiction action film directed by Peter Berg. 
</Backtrack>
\end{tcolorbox}
\begin{tcolorbox}[colback=green!15!white,colframe=white,boxsep=0pt,left=0pt,right=0pt,top=0pt,bottom=0pt,after=\vspace{-8pt} ]
<Search>
[Wiki\_RAG]: Search for detailed information about the \textbf{2012 military} sci-fi action film \textbf{Peter Berg}
</Search>
\end{tcolorbox}
\begin{tcolorbox}
[colback=gray!12!white,colframe=white,boxsep=0pt,left=0pt,right=0pt,top=0pt,bottom=0pt,after=\vspace{-8pt} ]
<Observation> 
'Wu Junyong (; born 1978) is a Chinese artist born in Putian, a city in Fujian province, China...'
</Observation>
\end{tcolorbox}
\begin{tcolorbox}[colback=violet!12!white,colframe=white,boxsep=0pt,left=0pt,right=0pt,top=0pt,bottom=0pt,after=\vspace{-8pt} ]
<Backtrack>  
The tool didn't provide the information about the 2012 military science fiction action film directed by Peter Berg. I will now try searching for that specific film.
</Backtrack> 
\end{tcolorbox}
\begin{tcolorbox}[colback=green!15!white,colframe=white,boxsep=0pt,left=0pt,right=0pt,top=0pt,bottom=0pt,after=\vspace{-8pt} ]
<Search>  
[Wiki\_RAG]:\textbf{ Peter Berg movies 2012} 
</Search> 
\end{tcolorbox}

\begin{tcolorbox}
[colback=gray!12!white,colframe=white,boxsep=0pt,left=0pt,right=0pt,top=0pt,bottom=0pt,after=\vspace{-8pt} ]
<Observation> 
'1. House Party -1991 - drama 2. Bloody Birthday -1992 - horror 3. Dumb $\&$ Dumber -1994 - comedy 4. The Rundown -2003 - action 5. Patriots Day -2016 - drama The Hidden -1999 - horror Fast $\&$ Furious 6 -2013 - action'
</Observation>
\end{tcolorbox}

\begin{tcolorbox}[colback=violet!12!white,colframe=white,boxsep=0pt,left=0pt,right=0pt,top=0pt,bottom=0pt,after=\vspace{-8pt} ]
<Backtrack>  
From the previous searches, \textbf{no information about a 2012 military science fiction action film} directed by Peter Berg and which Jason Snell worked on was retrieved. This information needs to be carefully re-evaluated and corrected.
</Backtrack> 
\end{tcolorbox}

\begin{tcolorbox}[colback=green!15!white,colframe=white,boxsep=0pt,left=0pt,right=0pt,top=0pt,bottom=0pt,after=\vspace{-8pt} ]
<Search>  
[Wiki\_RAG]: 2012 military science fiction action film directed by Peter Berg Jason Snell worked
</Search> 
\end{tcolorbox}

\begin{tcolorbox}
[colback=gray!12!white,colframe=white,boxsep=0pt,left=0pt,right=0pt,top=0pt,bottom=0pt,after=\vspace{-8pt} ]
<Observation> 
'1. Peter Berg (born March 11, 1964) is an American director, producer, writer, and actor...
\textbf{the military science fiction war film Battleship (2012)}...'
</Observation>
\end{tcolorbox}

\begin{tcolorbox}[colback=violet!12!white,colframe=white,boxsep=0pt,left=0pt,right=0pt,top=0pt,bottom=0pt,after=\vspace{-8pt} ]
<Think>  
From the search result, \textbf{I see that the 2012 movie directed by Peter Berg is "Battleship"}. However, the query mentions Jason Snell worked in it. Let me \textbf{verify} if Jason Snell indeed played a role in this film.
</Think> 
\end{tcolorbox}

\begin{tcolorbox}[colback=green!15!white,colframe=white,boxsep=0pt,left=0pt,right=0pt,top=0pt,bottom=0pt,after=\vspace{-8pt} ]
<Search>  
[Wiki\_RAG]: Jason Snell in Battleship
</Search> 
\end{tcolorbox}

\begin{tcolorbox}
[colback=gray!12!white,colframe=white,boxsep=0pt,left=0pt,right=0pt,top=0pt,bottom=0pt,after=\vspace{-8pt} ]
<Observation> 
'1. Jason H. Snell is an American special... 
Indiana Jones and the Kingdom of the \textbf{Battleship (2012)}...'
</Observation>
\end{tcolorbox}

\begin{tcolorbox}[colback=violet!12!white,colframe=white,boxsep=0pt,left=0pt,right=0pt,top=0pt,bottom=0pt,after=\vspace{-8pt} ]
<Summary> 
\textbf{After verifying with the search results, it is clear that "Battleship" (2012)} is the 2012 military science fiction action film directed by Peter Berg and in which Jason Snell worked. He played a role in this film.
</Summary>
\end{tcolorbox}

\begin{tcolorbox}[colback=yellow!15!white,colframe=white,boxsep=0pt,left=0pt,right=0pt,top=0pt,bottom=0pt ,after=\vspace{1pt}]%
<Conclusion> 
The 2012 military science fiction action film directed by Peter Berg in which Jason Snell worked is \textbf{"Battleship"}.
</Conclusion>
\end{tcolorbox}
\\ \hline
\end{tabular}
\caption{\label{table:case_study_on_wiki} A Case Study of a Wiki Problem on \M.}
\label{tab:case2}
\end{table*}

\section{Agentic Hacking Theoretical Analysis}
\label{sec:proof of agentic}
\subsection{Outcome-Only Reward Bias in Multi-Step Reasoning}
\label{sec:outcome_reward_multistep}

\begin{theorem}[Error Amplification under Outcome-Only Rewards]
\label{thm:multistep_outcome_bias}
Consider a multi-step reasoning trajectory
$y = (a_1, a_2, \ldots, a_T)$ generated by a policy
$\pi_{\boldsymbol{\theta}}(a_t \mid s_t)$ conditioned on an initial query $q_0$.
Let $r_{\mathrm{out}}$ denote an outcome-level reward assigned solely
based on the correctness of the final answer.
Under \textbf{outcome-only reward assignment without action-level reward decomposition},
each reasoning step $a_t$ receives an identical implicit reward
$\tilde{r}_t = \frac{r_{\mathrm{out}}}{T}$,
regardless of its semantic correctness\footnote{Here we ignore the group-relative advantage normalization for simplicity.}.
For any \textbf{erroneous or noisy step} $a_e$ with true marginal contribution
$\phi_e < 0$, the policy model parameter $\boldsymbol{\theta}_k$ at step $k$ update satisfies that:
\begin{equation}
\label{eq:multistep_bias_main}
\boldsymbol{\theta}_{k+1}
=
\boldsymbol{\theta}_k
+
\mu \cdot \frac{r_{\mathrm{out}}}{T} \cdot
\nabla_{\boldsymbol{\theta}}
\log \pi_{\boldsymbol{\theta}}(a_e \mid s_e),
\end{equation}
where $\mu > 0$ is the learning rate, $\log \pi_{\boldsymbol{\theta}}(a_e \mid s_e)$ is the output logits.
Consequently, the generation probability of $a_e$ increases whenever
$r_{\mathrm{out}} > 0$, irrespective of its negative contribution.
\end{theorem}

\begin{proof}
Under the RL estimator with outcome-only rewards, the policy gradient for a trajectory $y$ is:
\begin{equation}
\label{eq:multistep_reinforce}
\nabla_{\boldsymbol{\theta}} \mathbb{E}[R_{y}]
=
\mathbb{E}\!\left[
r_{\mathrm{out}}
\sum_{t=1}^{T}
\nabla_{\boldsymbol{\theta}}
\log \pi_{\boldsymbol{\theta}}(a_t \mid s_t)
\right].
\end{equation}
By linearity of expectation, this is equivalent to assigning each step
an identical implicit reward
$\tilde{r}_t = {r_{\mathrm{out}}}/{T}$:
\begin{equation}
\label{eq:uniform_step_reward}
\nabla_{\boldsymbol{\theta}} \mathbb{E}[R_{y}]
=
T \sum_{t=1}^{T}
\mathbb{E}\!\left[
\tilde{r}_t \cdot
\nabla_{\boldsymbol{\theta}}
\log \pi_{\boldsymbol{\theta}}(a_t \mid s_t)
\right].
\end{equation}
This formulation implicitly assumes uniform contribution across steps.
In practice, however, reasoning steps exhibit heterogeneous marginal
contributions $\{\phi_t\}_{t=1}^{T}$, where
\(
\sum_{t=1}^{T} \phi_t = r_{\mathrm{out}},
\)
and some $\phi_t$ \textbf{may be negative due to hallucination,
spurious inference, or retrieval noise}.
Thus, consider an erroneous step $a_e$ with $\phi_e < 0$.
The induced parameter update is
\begin{equation}
\label{eq:multistep_param_update}
\Delta \boldsymbol{\theta}
=
\mu \cdot \frac{r_{\mathrm{out}}}{T} \cdot
\nabla_{\boldsymbol{\theta}}
\log \pi_{\boldsymbol{\theta}_k}(a_e \mid s_e).
\end{equation}
Applying first-order Taylor expansion to step update:
\begin{equation}
\footnotesize
\begin{aligned}
\label{eq:multistep_taylor}
&\pi_{\boldsymbol{\theta}_{k+1}}(a_e)
=
\pi_{\boldsymbol{\theta}_k}(a_e)
+
\nabla_{\boldsymbol{\theta}}
\pi_{\boldsymbol{\theta}_k}(a_e)
\cdot \Delta \boldsymbol{\theta}
+ \mathcal{O}(\mu^2) \nonumber\\
&=
\pi_{\boldsymbol{\theta}_k}(a_e)
\Bigg(
1
+
\mu \cdot \frac{r_{\mathrm{out}}}{T}
\Big\|
\nabla_{\boldsymbol{\theta}}
\log \pi_{\boldsymbol{\theta}_k}(a_e \mid s_e)
\Big\|_2^2
\Bigg)
+ \mathcal{O}(\mu^2).
\end{aligned}
\end{equation}
where we used the log-derivative identity
$\nabla_{\boldsymbol{\theta}} \pi = \pi \nabla_{\boldsymbol{\theta}} \log \pi$.
Since $r_{\mathrm{out}} > 0$ for correct outcomes and gradient norms
are non-negative, the bracketed term exceeds unity, implying
\begin{equation}
\pi_{\boldsymbol{\theta}_{k+1}}(a_e)
>
\pi_{\boldsymbol{\theta}_k}(a_e)
\quad \text{regardless of } \phi_e < 0.
\end{equation}
To characterize long-term accumulation, define the
\textbf{step-level reward allocation error} as:
\begin{equation}
\label{eq:step_reward_error}
\varepsilon_e
=
\frac{r_{\mathrm{out}}}{T}
-
\phi_e.
\end{equation}
For erroneous steps, $\varepsilon_e > 0$ whenever $r_{\mathrm{out}} > 0$.
Over $K$ training iterations, the cumulative bias satisfies
\begin{equation}
\footnotesize
\label{eq:multistep_cumulative_bias}
\mathbb{E}\!\left[\|\bm{B}_e^{(K)}\|_2\right]
\simeq
\varepsilon_e \cdot K \cdot
\mathbb{E}\!\left[
\big\|
\nabla_{\boldsymbol{\theta}}
\log \pi_{\boldsymbol{\theta}}(a_e \mid s_e)
\big\|_2
\right],
\end{equation}
demonstrating monotonic amplification of erroneous reasoning steps (\textbf{memory-aware actions}) in the absence of backtracking or summarizing.
\end{proof}

\subsection{Bias Isolation via Pop-Based Stack and Context Masking}
\label{sec:pop_stack_bias_isolation}
\begin{theorem}[Structural Isolation via Pop-Based Backtracking]
\label{thm:pop_stack_isolation}
Consider a multi-step reasoning trajectory
$y = (a_1, a_2, \ldots, a_T)$ generated by a policy
$\pi_{\boldsymbol{\theta}}(a_t \mid s_t)$.
Suppose that at step $e$, an erroneous action $a_e$ is detected and removed
via a pop-based stack mechanism, producing a pruned trajectory
$\tilde{y}$.
The removal is implemented by masking $a_e$ from the conditioning context
of all subsequent actions.
Under this counterfactual trajectory $\tilde{y}$,
the expected policy gradient update contains no contribution from $a_e$:
\begin{equation}
\mathbb{E}\!\left[
\nabla_{\boldsymbol{\theta}}
\log \pi_{\boldsymbol{\theta}}(a_e \mid s_e)
\cdot r_{\mathrm{out}}
\right]
= 0.
\end{equation}
Thus, erroneous reasoning steps do not receive outcome-level rewards
and are not systematically reinforced.
\end{theorem}

\begin{proof}
Under outcome-only rewards, the standard advantage estimator is
\begin{equation}
\nabla_{\boldsymbol{\theta}} \mathbb{E}[R_{y}]
=
\mathbb{E}\!\left[
r_{\mathrm{out}}
\sum_{t=1}^{T}
\nabla_{\boldsymbol{\theta}}
\log \pi_{\boldsymbol{\theta}}(a_t \mid s_t)
\right].
\end{equation}

\paragraph{Pop-based backtracking and context masking.}
When $a_e$ is removed, all subsequent actions are generated or evaluated
under a masked context in which $a_e$ is excluded. The resulting trajectory is
\begin{equation}
\tilde{y}
=
(a_1, \ldots, a_{e-1}, a_{e+1}', \ldots, a_T').
\end{equation}
The attention mask ensures that $a_e$ does not contribute to hidden states
or logits involved in producing the final answer, effectively addressing the problem of reward hacking from outcome reward to be allocated the harmless popped actions.

\paragraph{Counterfactual independence.}
Under the masked context, the outcome distribution satisfies
\begin{equation}
p_{\boldsymbol{\theta}}(a_T \mid a_1, \ldots, a_{e-1}, \cancel{a_e}, \ldots)
=
p_{\boldsymbol{\theta}}(a_T \mid \tilde{y}),
\end{equation}
implying that $r_{\mathrm{out}}$ depends only on $\tilde{y}$ and is
conditionally independent of $a_e$.

\paragraph{Vanishing gradient contribution.}
The policy gradient under $\tilde{y}$ becomes
\begin{equation}
\nabla_{\boldsymbol{\theta}} \mathbb{E}[R_{\tilde{y}}]
=
\mathbb{E}\!\left[
r_{\mathrm{out}}
\sum_{t \in \tilde{y}}
\nabla_{\boldsymbol{\theta}}
\log \pi_{\boldsymbol{\theta}}(a_t \mid s_t)
\right].
\end{equation}
Since $a_e \notin \tilde{y}$, the corresponding score-function term
never appears. Equivalently,
\begin{equation}
    \mathbb{E}\!\left[
\nabla_{\boldsymbol{\theta}}
\log \pi_{\boldsymbol{\theta}}(a_e \mid s_e)
\cdot r_{\mathrm{out}}
\right]
= 0.
\end{equation}
Therefore, erroneous reasoning steps are structurally isolated from
outcome-level rewards.
\end{proof}

\paragraph{Discussion.}
Rather than numerically redistributing rewards, the pop-based stack mechanism
corrects outcome-only reward bias through \emph{structural isolation}.
By explicitly removing erroneous actions from the conditioning context via
attention masking, policy gradients are evaluated on a counterfactual
trajectory in which these actions never occurred, eliminating the error
amplification effect identified in
Section~\ref{sec:outcome_reward_multistep}.

\section{More Method Details: Pop-based Attention Mask}
\label{sec:pop_attention_mask}
We extend the standard causal attention mechanism with a pop-based attention mask to handle dynamic context updates induced by \texttt{<summary>} and \texttt{<backtrack>} operations (\texttt{<pop>}). In autoregressive Transformers, causal attention ensures that each token at position ( $i$ ) can only attend to tokens at positions ( $j \leq i$ ). However, in our setting, causal ordering alone is insufficient, since tokens that appear earlier in the sequence may have been explicitly \textbf{removed from the active context} via a pop operation and should no longer be accessible.

Consider the example shown in Figur~\ref{fig:newmethod.png} Part 4, where the sequence follows an \textbf{Action–Summary} pattern: \textbf{Act. 1 → Act. 2 → Sum. 3 → Act. 4}. The summary token \textbf{Sum. 3} aggregates the information of \textbf{Act. 2}, after which \textbf{Act. 2} is popped from the active memory. As a result, when generating \textbf{Act. 4}, the model should be allowed to attend to \textbf{Sum. 3}, but must be prevented from directly attending to \textbf{Act. 2}, even though \textbf{Act. 2} appears earlier in the sequence. Thus, the attention mask between \textbf{Act. 2} and \textbf{Act. 4} should be False. This constraint is indicated as \textbf{Mask after POP} in the attention matrix.
Formally, let ( $M^{\text{causal}}$ ) denote the standard causal mask:
\begin{equation}
    M_{i,j}^{\text{causal}} =
\begin{cases}
0, & j \le i \\
-\infty, & j > i
\end{cases}
\end{equation}
We introduce an additional pop-based mask ( $M^{\text{pop}}$ ), where a token ( j ) is masked if it has been popped before step ( i ):
\begin{equation}
M_{i,j}^{\text{pop}} =
\begin{cases}
-\infty, & \text{if} j \text{ is popped before } i \\
0, & \text{otherwise}
\end{cases}
\end{equation}
The final attention mask is defined as:
\begin{equation}
M_{i,j} = M_{i,j}^{\text{causal}} + M_{i,j}^{\text{pop}}.
\end{equation}
This formulation ensures that once an action has been summarized and popped, its information is only accessible through the corresponding summary token in all subsequent steps.

\section{Broader Generalization}
\label{sec:broader_generalization}

To evaluate the effectiveness of \M across different application scenarios, we conduct additional experiments on MedQA, DeepResearch  Bench, ALFWorld, and WebShop, covering medical question answering, long-form report generation, and interactive decision making.

\subsection{Medical Question Answering.}
Recent work has demonstrated the potential of LLMs for domain-specific healthcare applications~\cite{xu2025dearllm,jiang2024tcrag,fang2026graphwalker}. 
To evaluate the generalization of \M in a domain-specific setting, we conduct experiments on MedQA~\cite{jin2020disease}. Table~\ref{tab:medqa_results} reports the performance of different methods using Qwen2.5-72B as the judge. Compared with the base model, ReAct improves format accuracy from 39\% to 56\% but reduces answer  accuracy from 84\% to 79\%. In contrast, \M achieves 92\% format accuracy and 87\% answer accuracy, outperforming both the base model and ReAct on both metrics. These results demonstrate that \M generalizes effectively to domain-specific medical question answering while improving format compliance without sacrificing answer quality.

\begin{table}[t]
\centering
\fontsize{8pt}{10pt}\selectfont
\renewcommand{\arraystretch}{1}
\setlength{\tabcolsep}{8pt}

\begin{tabular}{l|cc}
\toprule
\rowcolor{gray!10}
\textbf{Method} & \textbf{Format Acc (\%)} & \textbf{Answer Acc (\%)} \\
\midrule
Base   & 39\% & 84\% \\
ReAct  & 56\% & 79\% \\
\M     & 92\% & 87\% \\
\bottomrule
\end{tabular}
\caption{Results on the MedQA test set. The Qwen2.5-7B-Instruct model is evaluated using Qwen2.5-72B as the judge.}
\label{tab:medqa_results}
\end{table}

\subsection{Long-Form Research and Report Generation}
\begin{figure*}[htbp]
    \centering
    \begin{subfigure}{1\linewidth}
        \centering
        \includegraphics[width=1\linewidth]{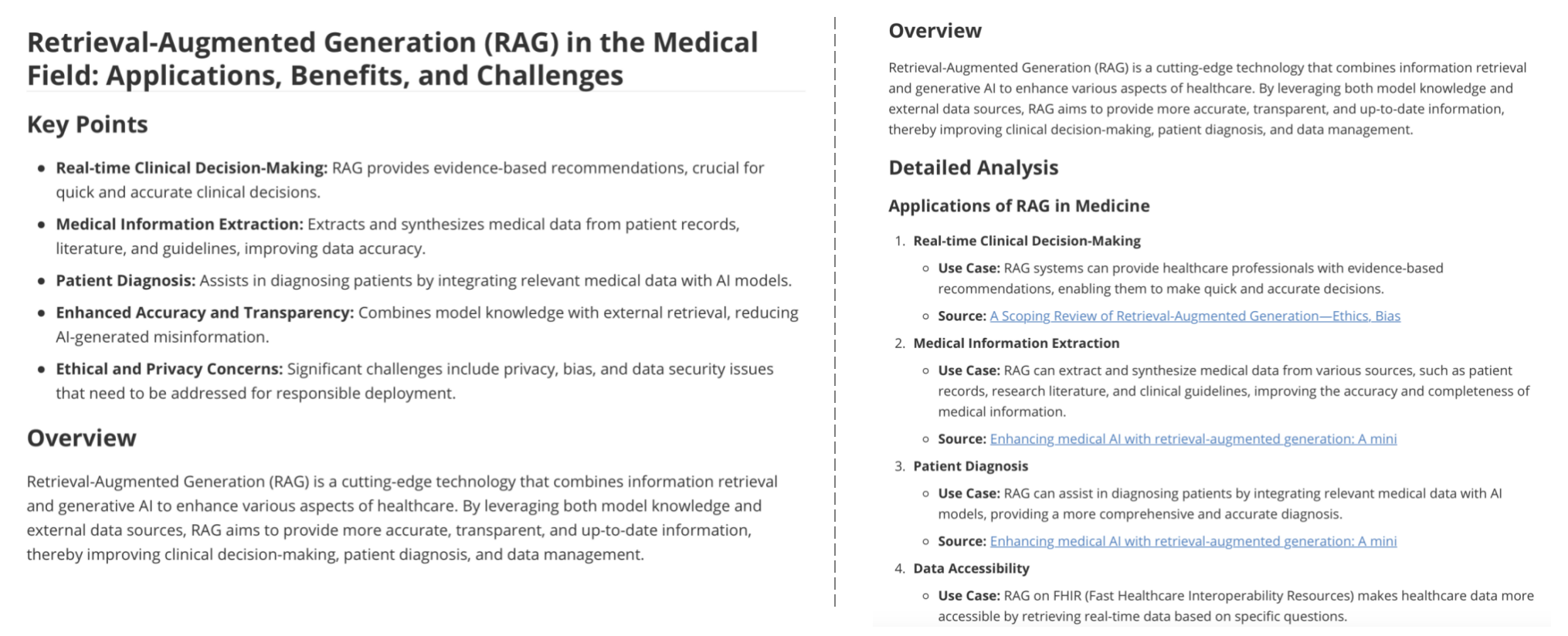}
        \caption{Report Part (1/2)}
        \label{subfig:report1}
    \end{subfigure}


    \begin{subfigure}{1\linewidth}
        \centering
        \includegraphics[width=1\linewidth]{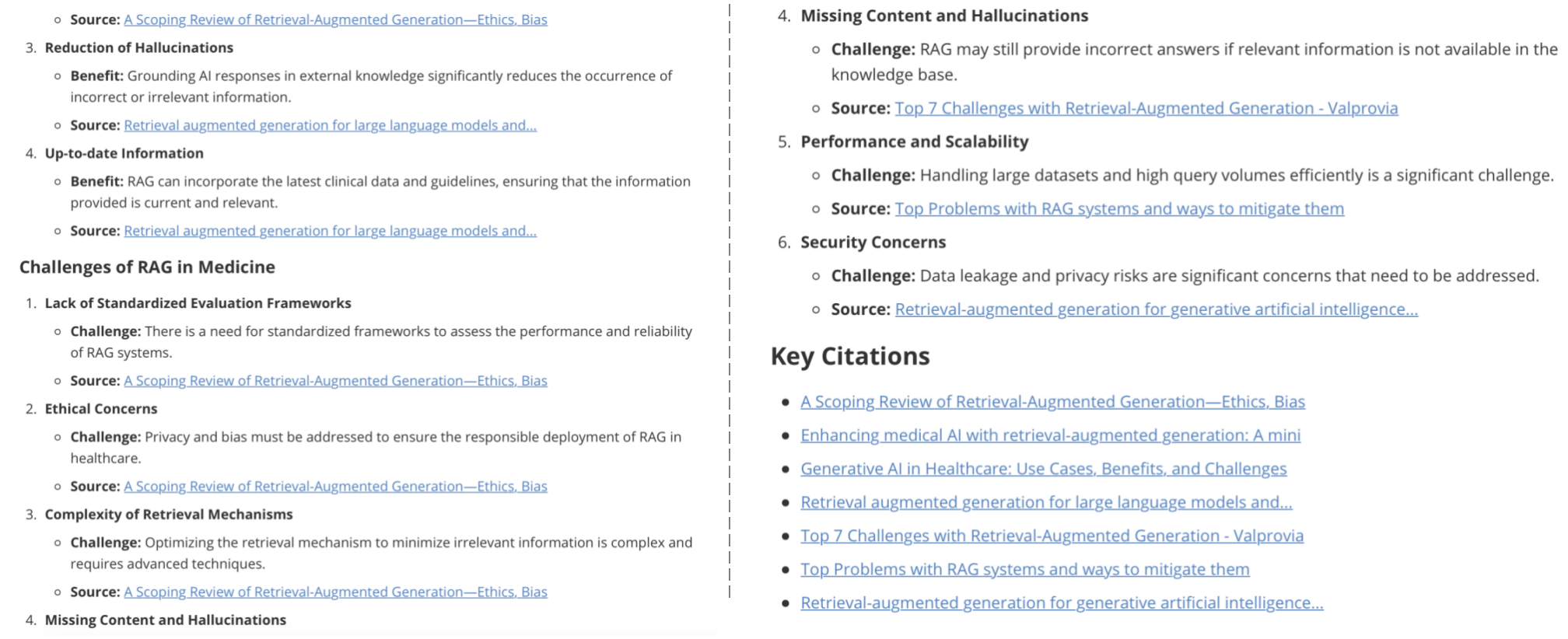}
        \caption{Report Part (2/2)}
        \label{subfig:report2}
    \end{subfigure}

    \caption{Report Generation for query "\textit{Give me a report of topic <RAG's Future in Medical Domain>}"}
    \label{fig:reporter_rag}
\end{figure*}

Figure~\ref{fig:reporter_rag} presents the report-generation performance of \M~ on the 7B model, inferred with the query:
\textit{``Give me a report of topic <RAG's Future in Medical Domain>.''}
As shown, after multiple rounds of retrieval, the model produces a high-quality report accompanied by well-aligned citations.
This result further demonstrates the strong generalization capability of our reinforcement learning–based agent training strategy.

To quantitatively evaluate the generalization of \M~to long-horizon agentic scenarios beyond QA, we benchmark on DeepResearch Bench~\cite{du2025deepresearch}, a multi-document report generation task that requires sustained multi-step retrieval, information aggregation, and structured composition. We evaluate using RACE (report quality) and FACT (factual consistency, including effective citation accuracy and claim-level accuracy).
As shown in Table~\ref{tab:deepresearch}, \M~substantially outperforms all baselines, including RAG, ReACT, and EDR~\cite{prabhakar2025enterprise}, achieving the highest RACE score (43.60) and a significant margin on FACT metrics (+8.14 on Eff.c.\ and +8.56 on C.acc.\ over the best baseline).
\begin{table}[h]
\centering
\small
\begin{tabular}{l|ccc}
\toprule
\textbf{Method} & \textbf{RACE} & \textbf{FACT (Eff.c.)} & \textbf{FACT (C.acc.)} \\
\midrule
RAG & 36.68 & 2.93 & 10.86 \\
ReACT & 42.21 & 4.69 & 23.78 \\
EDR & 40.99 & 12.28 & 23.70 \\
Ours & \textbf{43.60} & \textbf{20.42} & \textbf{32.34} \\
\bottomrule
\end{tabular}
\caption{Results on DeepResearch Bench for multi-document report generation.}
\label{tab:deepresearch}
\end{table}

\subsection{Interactive Decision-Making Environments.}

We further evaluate \M on interactive long-horizon decision making with ALFWorld~\cite{ALFWorld20} and WebShop~\cite{yao2022webshop}. These environments require agents to maintain state, incorporate feedback, and adapt strategies over extended trajectories. \M improves ALFWorld success rate from 36.5\% to 51.1\%, WebShop task score from 14.7 to 27.8, and WebShop success rate from 0.8\% to 2.3\%. Despite category-level variations on ALFWorld, these results demonstrate that \M generalizes beyond static QA to sequential interactive environments.

\begin{table*}[t]
\centering
\footnotesize
\setlength{\tabcolsep}{6pt}
\renewcommand{\arraystretch}{1.12}
\resizebox{0.7\textwidth}{!}{%
\begin{tabular}{l | c c c c c c c | c c}
\toprule
\rowcolor{gray!30}
\textbf{Method} & \multicolumn{7}{c|}{\textbf{ALFWorld}} &
\multicolumn{2}{c}{\textbf{WebShop}} \\
\rowcolor{gray!30}
& \textbf{Pick} & \textbf{Look} & \textbf{Clean} &
\textbf{Heat} & \textbf{Cool} & \textbf{Pick2} & \textbf{Overall} &
\textbf{Score} & \textbf{SR} \\
\midrule
Prompt & 45.8 & 48.4 & \textbf{34.5} & \textbf{41.0} &
19.6 & 31.7 & 36.5 & 14.7 & 0.8 \\
\M & \textbf{84.7} & \textbf{80.6} & 24.1 & 33.3 &
\textbf{37.0} & \textbf{51.2} & \textbf{51.1} &
\textbf{27.8} & \textbf{2.3} \\
\bottomrule
\end{tabular}}
\caption{Generalization of \M to interactive decision-making
environments. ALFWorld reports task-specific and overall success rates
(\%); WebShop reports the official task score and success rate
(SR, \%).}
\label{tab:interactive_generalization}
\end{table*}

\section{Computational Resources and Software Environment}
Experiments were carried out on a server equipped with two Intel Xeon Platinum 8488C processors, comprising 96 physical cores and 192 hardware threads in total. 
The system was configured with two NVIDIA A800 GPUs, each providing 80~GB of GPU memory, and 503~GB of system RAM. 
All experiments were conducted on Ubuntu~22.04.5~LTS.
The software environment was based on Python~3.11.11, with package management handled through Conda (version~23.5.2). 
Model implementation and training were performed using PyTorch~2.6.0 with CUDA support, together with HuggingFace Transformers~4.51.3 and SpaCy~3.8.4. 
Unless otherwise specified, all models and libraries were used with their default parameter settings.
The total training time for our method was approximately 21{,}582~seconds under the above hardware configuration.
The computational cost of training, data preprocessing, and inference varied depending on the experimental configuration and dataset scale, and all reported results were obtained under the same hardware and software setup.


\section{The Use of Large Language Model}
In this study, generative AI software tools were used solely as auxiliary support for language refinement, programming-related assistance, and figure visualization. Specifically, large language models were employed to improve grammatical accuracy, readability, and overall writing quality, as well as to provide general coding suggestions or debugging guidance. In addition, AI-based image generation tools were used to enhance the visual appearance of figure \ref{fig:intro}; however, the figure design, initial draft, and final verification were entirely performed by the authors. All AI-assisted outputs were carefully reviewed and validated by the authors prior to inclusion. The research conception, methodological design, and analysis of experimental results were conducted exclusively by the authors, and generative AI tools did not contribute to the formulation of research ideas or the derivation of conclusions.

\end{document}